\documentclass[12pt]{article}
\usepackage{amsmath}
\usepackage{jheppub}
\usepackage{amssymb,amsfonts}
\usepackage{latexsym}
\usepackage{mathtools}
\usepackage{tikz-cd}
\usepackage{tikz}
\usepackage{tikz-3dplot}
\usetikzlibrary{knots}
\usetikzlibrary{calc}
\usetikzlibrary{topaths}
\usetikzlibrary{decorations}
\usetikzlibrary{decorations.pathmorphing}
\usetikzlibrary{arrows,decorations.markings,cd}
\usetikzlibrary{calc,arrows,cd,decorations.markings,snakes}

\usepackage{enumerate}
\usepackage{empheq}
\usepackage{dsfont}
\relax

\def\be{\begin{equation}}
	\def\ee{\end{equation}}
\def\bea{\begin{eqnarray}}
	\def\eea{\end{eqnarray}}

\renewcommand{\vec}[1]{\boldsymbol{#1}}

\makeatletter
\DeclareFontFamily{OMX}{MnSymbolE}{}
\DeclareSymbolFont{MnLargeSymbols}{OMX}{MnSymbolE}{m}{n}
\SetSymbolFont{MnLargeSymbols}{bold}{OMX}{MnSymbolE}{b}{n}
\DeclareFontShape{OMX}{MnSymbolE}{m}{n}{
	<-6>  MnSymbolE5
	<6-7>  MnSymbolE6
	<7-8>  MnSymbolE7
	<8-9>  MnSymbolE8
	<9-10> MnSymbolE9
	<10-12> MnSymbolE10
	<12->   MnSymbolE12
}{}
\DeclareFontShape{OMX}{MnSymbolE}{b}{n}{
	<-6>  MnSymbolE-Bold5
	<6-7>  MnSymbolE-Bold6
	<7-8>  MnSymbolE-Bold7
	<8-9>  MnSymbolE-Bold8
	<9-10> MnSymbolE-Bold9
	<10-12> MnSymbolE-Bold10
	<12->   MnSymbolE-Bold12
}{}

\let\llangle\@undefined
\let\rrangle\@undefined
\DeclareMathDelimiter{\llangle}{\mathopen}%
{MnLargeSymbols}{'164}{MnLargeSymbols}{'164}
\DeclareMathDelimiter{\rrangle}{\mathclose}%
{MnLargeSymbols}{'171}{MnLargeSymbols}{'171}
\makeatother

\usepackage{amsthm}
\newtheorem{definition}{Definition}

\newtheorem{theorem}{Theorem}
\newtheorem{lemma}{Lemma}
\newtheorem{proposition}{Proposition}

\newcommand{\id}{\mathds{1}}

\newcommand{\End}{\text{End }}

\newcommand{\g}{\mathsf g}
\newcommand{\W}{\mathcal W}
\newcommand{\q}{\mathfrak q}
\newcommand{\G}{\mathsf G}
\renewcommand{\O}{\mathsf O}
\newcommand{\Q}{\mathsf Q}
\newcommand{\K}{\mathsf K}
\newcommand{\sG}{\mathcal G}
\newcommand{\PI}[1]{\mathcal Z\left[#1\right]}
\renewcommand{\P}{\mathcal P}
\renewcommand{\H}{\mathsf H}

\newcommand{\avg}[1]{\Big\llangle #1\Big\rrangle_{Sp}}

\newcommand{\Aut}{\text{Aut}}
\newcommand{\tr}{\text{tr}}

\newcommand{\what}{\widehat}
\newcommand{\Sp}{\operatorname{Sp}}
\renewcommand{\End}{\operatorname{End}}
\renewcommand{\t}{\lambda}

\tikzset{
	pics/torus/.style n args={3}{
		code = {
			\providecolor{pgffillcolor}{rgb}{1,1,1}
			\begin{scope}[
				yscale=cos(#3),
				outer torus/.style = {draw,line width/.expanded={\the\dimexpr2\pgflinewidth+#2*2},line join=round},
				inner torus/.style = {draw=pgffillcolor,line width={#2*2}}
				]
				\draw[outer torus] circle(#1);\draw[inner torus] circle(#1);
				\draw[outer torus] (180:#1) arc (180:360:#1);\draw[inner torus,line cap=round] (180:#1) arc (180:360:#1);
			\end{scope}
		}
	}
}

\newcommand{\hb}[1]{\vcenter{\hbox{\begin{tikzpicture}[fill=none,draw=black]
				\pic{torus={1cm}{2.8mm}{70}};
				\ifthenelse{\equal{#1}{0}}{}{\draw [yscale=cos(70),black,thick](1.05,0.08) arc (0:285:1.05cm);
				\draw [yscale=cos(70),black,thick](1.05,0.08) arc (0:-45:1.05cm);
				\node [black] at (0.5,-0.3) {$#1$};}
\end{tikzpicture}}}}

\newcommand{\weyl}[2]{\vcenter{\hbox{\begin{tikzpicture}
			\draw[thick] (0,0) to (0,2);
			\draw[thick] (0,2) to (2,2);
			\draw[thick] (2,2) to (2,0);
			\draw[thick] (2,0) to (0,0);
			\draw[blue] (1,0) to (1,2);
			\node[circle,inner sep=2pt,draw, fill, color = white] at (1,1) {};
			\draw[blue] (0,1) to (2,1);
			\node[anchor = south] at (0.4,1) {$#1$};
			\node[anchor = south] at (1.3,0) {$#2$};
\end{tikzpicture}}}}

\newcommand{\defect}[1]{\vcenter{\hbox{\begin{tikzpicture}[fill=none,draw=black]
				\pic{torus={1cm}{2.8mm}{70}};
				\ifthenelse{\equal{#1}{0}}{}{\draw [dashed,yscale=cos(70),black,thick](1.05,0.08) arc (0:285:1.05cm);
					\draw [dashed,yscale=cos(70),black,thick](1.05,0.08) arc (0:-45:1.05cm);
					\node [black] at (0.5,-0.3) {$#1$};}
\end{tikzpicture}}}}

\title{Abelian 3D TQFT gravity, ensemble holography and stabilizer states}

\author{Nikolaos Angelinos}
\affiliation{Yau Mathematical Sciences Center, Tsinghua University, Beijing 100084, China} 
\abstract{We construct a model of 3D quantum gravity based on abelian topological quantum field theory (TQFT), by defining the gravitational path-integral as a sum over all 3D topologies with genus-$g$ boundary $\Sigma_g$. The path-integral of an abelian TQFT $\mathcal T$ on any single topology with boundary $\Sigma_g$ prepares a stabilizer state. This way, $\mathcal T$ partitions all these topologies into finitely many equivalence classes, where each topology within a class is associated with the same stabilizer state. The gravitational path-integral can thus be rephrased as a weighted sum over representative topologies, which are further organized into orbits under the mapping class group of $\Sigma_g$. One orbit is represented by handlebodies, whose average reproduces the “Poincaré series of the vacuum”, while additional orbits describe non-handlebody topologies.
	The resulting quantum gravity state is $\Sp(2g,\mathbb Z)$-invariant and can be expressed as a weighted average of 2D CFT partition functions on $\Sigma_g$. This establishes a duality between a weighted sum over bulk topologies and a weighted sum over boundary CFTs. We introduce the “$\t$-matrix”, which relates bulk and boundary weights. The $\t$-matrix can be fully determined by the set of topological boundary conditions that the TQFT admits, and we present a systematic procedure to construct this set. Using this framework, we evaluate the $\t$-matrix and the TQFT gravity state in several tractable examples.}
	
	\makeatletter
	\gdef\@fpheader{}
	\makeatother
\begin{document}
	\maketitle

\section{Introduction}

The quantum gravity path-integral is notoriously difficult to define, as it requires summing over all possible spacetime topologies and integrating over all metrics compatible with each topology. In models of quantum gravity based on a topological quantum field theory (TQFT), this problem simplifies, as the theory is not sensitive to the metric, leaving only a sum over topologies. Even so, a general prescription for the quantum gravity path-integral remains elusive.

A concrete proposal was made by Maloney and Witten \cite{Maloney_2010}, who defined the AdS$_3$
gravity path integral as a sum over handlebody topologies. Their result lacked a clear holographic interpretation, since the putative dual 2D CFT would have a continuous spectrum with negative degeneracies. A similar construction was carried out in 3D abelian Chern–Simons (CS) theory \cite{Maloney_2020,Afkhami-Jeddi:2020ezh}, where the sum over handlebodies yielded an expression that can be interpreted as an average over the Narain moduli space of 2D CFTs. This motivated the idea of ``ensemble holography" in three dimensions \cite{Collier:2021rsn,Cotler:2020hgz,Benjamin:2021wzr,Ashwinkumar:2023jtz,Ashwinkumar:2023ctt,Kames-King:2023fpa,Forste:2024zjt,Perez:2020klz,Datta:2021ftn,Benjamin:2021ygh,Chakraborty:2021gzh,Raeymaekers:2021ypf,Benini:2022hzx,Saidi:2024zdj}, where the bulk gravitational path integral is dual not to a single boundary CFT, but rather to an ensemble of CFTs.

Ensemble holography in 3D has also been explored in settings where the boundary theories are RCFTs \cite{Castro:2011zq,Jian:2019ubz,Romaidis:2023zpx,Meruliya:2021utr,Meruliya:2021lul}. However, except for some simple cases, the handlebody sum does not produce a consistent RCFT ensemble. Abelian 3D TQFTs provide more tractable models, but even in this simplified setting, restricting to genus-$g$ handlebodies typically results in a boundary ensemble with genus-dependent weights, which is incompatible with holography.  

A consistent prescription was introduced in \cite{Dymarsky:2024frx}, using a genus-reduction map applied to a sum over large-genus handlebodies. In this approach, the boundary CFT ensemble has genus-independent coefficients, and the bulk sum (after genus reduction) includes contributions from all topologies, not just handlebodies.

 The same work showed that the abelian TQFT path integral on any 3D topology with a genus-$g$ boundary $\Sigma_g$ prepares a stabilizer state associated with a maximally isotropic subgroup of $H_1(\Sigma_g,\mathbb Z_N)$ with respect to the intersection form, where $N$ is the exponent of the $1$-form symmetry group $\G$. Consequently, for a given TQFT $\mathcal T$, all topologies with boundary $\Sigma_g$ are partitioned into finitely many equivalence classes, with the path-integral of $\mathcal T$ preparing the same state (up to normalization) within each class. Choosing a representative topology in each class, the gravitational path integral can be reformulated as a weighted sum over these representatives\footnote{Also see \cite{Nicosanti:2025xwu} for a similar approach, where the quantum gravity path integral for abelian 3D TQFT was rephrased as a sum over homology groups.}, which are further organized into orbits under the mapping class group $\Sp(2g,\mathbb Z)$ of $\Sigma_g$. We do not attempt a first-principles derivation of the weights and keep them arbitrary. Compatibility with holography requires $\Sp(2g,\mathbb Z)$-invariance, so the weights must be uniform within each orbit. We therefore define the gravitational path integral of a TQFT $\mathcal T$ as
 \be \mathcal Z_{gravity}[\mathcal T,\Sigma_g;w]\equiv\sum_{\vec d}w_{\vec d}\avg{|\Omega_{\vec d}\rangle},\label{Zdef1}\ee
where $w={w_{\vec d}}$ are the \textit{bulk weights}, $|\Omega_{\vec d}\rangle$ is the stabilizer state prepared by the path-integral on the topology $\mathcal V_g^{\vec d}$ with $\partial\mathcal V_g^{\vec d}\cong\Sigma_g$, $\vec d$ labels its first homology group, and $\avg{\cdot}$ denotes a uniform average over the $\Sp(2g,\mathbb Z)$-orbit, which contains finitely many states. The classification of the representative topologies will be given in subsection \ref{etaLagr}.


The resulting quantum gravity state is $\Sp(2g,\mathbb Z)$-invariant and can be expressed as a linear combination of stabilizer states specified by subgroups of $\G$ that are maximally isotropic with respect to the quadratic form $\q$. These states are organized into orbits under the orthogonal group $\O(\G)$ preserving $\q$. Each such state corresponds to a topological boundary condition (TBC) \cite{KS1} on $\Sigma_g$ and is associated with a bosonic 2D CFT partition function \cite{Barbar:2023ncl}. This yields a dual description as a weighted ensemble of 2D CFTs,
\be \mathcal Z_{gravity}[\mathcal T,\Sigma_g;w]=\sum_{\vec q}m_{\vec q}{1\over |\O_{\vec q}|}\sum_{i\in\O_{\vec q}} Z^{(g)}_{\vec q,i},\ee
where $Z_{\vec q,i}^{(g)}$ is a CFT partition function on $\Sigma_g$. The outer sum runs over 
$\O(\G)$-orbits, the inner sum averages uniformly within each orbit $\O_{\vec q}$, and the coefficients $m_{\vec q}$ are determined by $w_{\vec d}$. We will refer to $m_{\vec q}$ as \textit{boundary weights}.

Our main result is introducing the ``$\t$-matrix'' and a method to compute it, relating bulk and boundary weights,
\be m_{\vec q}=\sum_{\vec d} w_{\vec d}\t_{\vec d,\vec q}.\ee
The $\t$-matrix is completely determined by the set of TBCs that the TQFT admits. Concretely, it can be written as
\be\t_{\vec d,\vec q}=\sum_{\vec q'}v_{\vec q'}^{\vec d}Y_{\vec q';\vec q},\ee
where the matrix $Y_{\vec q';\vec q}$ is obtained from the overlap matrix of TBC states and $v_{\vec q}^{\vec d}$ is a universal matrix determined only by the $1$-form symmetry group $\G$ and the genus of the boundary surface.
 
 
We compute the $\lambda$-matrix and the quantum gravity state in various examples by a direct method, which requires identifying the full set of TBCs. For a general abelian TQFT, this is a highly nontrivial problem. In section \ref{sec3} we present a systematic procedure for constructing the full set of TBCs in a class of abelian TQFTs, based on the groupoid structure introduced in \cite{Gaiotto_2021}.

The paper is organized as follows. Section \ref{sec1} reviews the relation between abelian TQFTs and the quantum stabilizer formalism. Section \ref{sec2} classifies the two types of stabilizer states that appear on the bulk and boundary sides of the duality. Section \ref{sec:TQFTgrav} introduces the $\t$-matrix, relating the bulk and boundary weights and describes a general method to compute it. Section \ref{sec3} develops additional tools that we use to construct the complete set of TBCs. Section \ref{sec:examples} explicitly computes the $\t$-matrix and gravity states in various examples. We conclude in section \ref{sec:conc} with a summary and open problems.

\section{Abelian TQFT and Stabilizer Formalism}\label{sec1}

Our starting point is an Abelian 3D TQFT $\mathcal T$ specified by the pair $(\G,\q)$, where $\G$ is a finite abelian group and $\q$ is a non-degenerate quadratic form on $\G$. This pair fully determines the data of a pointed Modular Tensor Category (MTC) whose objects are anyons labeled by elements of $\G$ and their topological spins are given by
\be
\theta(a)=\exp\!\left(\frac{\pi i}{N}\,\q(a)\right), \qquad a\in\G,
\ee
where $N$ is the exponent of $\G$ and $\q:\G\to\mathbb Z_{2N}$ is the quadratic form.
The Frobenius–Schur exponent $N_{FS}$ is the minimal integer such that $\theta(a)^{N_{FS}}=1$ for all $a\in\G$.
To avoid subtleties when $N$ is even, we will henceforth assume $N$ is odd, so that $N_{FS}=N$.
Introducing the primitive root of unity $\omega=e^{2\pi i/N}$, the spins can be written compactly as
\be
\theta(a)=\omega^{2^{-1}\q(a)} ,
\ee
where $2^{-1}$ denotes the multiplicative inverse of $2$ in $\mathbb Z_N$.
The spins determine the braiding matrix
\be
\sqrt{|\G|}\,S_{ab}=\frac{\theta(a+b)}{\theta(a)\theta(b)}=\omega^{\g(a,b)},
\ee
where $\g:\G\times\G\to\mathbb Z_N$ is the bilinear form associated to $\q$
\be \g(a,b)=2^{-1}(\q(a+b)-\q(a)-\q(b))\mod N.\ee
The spins also determine the chiral central charge $c$ modulo $8$ via
\be
e^{\frac{2\pi i}{8}c}=\frac{\sum_{a\in\G}\theta(a)}{|\sum_{a\in\G}\theta(a)|}.
\ee
 Throughout, we will only consider $(\G,\q)$ such that the chiral central charge vanishes. This condition is necessary for the theory to admit topological boundary conditions \cite{Kaidi:2021gbs}.


Any finite abelian group $\G$ decomposes into a product of $p$-groups,
\be
\G\cong \G_{p_1}\times\cdots\times\G_{p_k},
\ee
where $p_1,\dots,p_k$ are distinct primes, and $\G_{p}$ denotes a group consisting of elements whose order is a power of $p$. The quadratic form factorizes accordingly,
\be
\q=\q_{p_1}\oplus\cdots\oplus\q_{p_k}, \qquad 
\q_{p_i}=\q|_{\G_{p_i}} .
\ee
Thus the TQFT itself decomposes into a product of decoupled TQFTs,  \cite{Kaidi:2021gbs}
\be
\mathcal T\;\cong\; \mathcal T_{p_1}\times\cdots\times\mathcal T_{p_k}, 
\qquad 
\mathcal T_{p_i}=(\G_{p_i},\q_{p_i}).
\ee
Without loss of generality, we can restrict to the case where $\G$ is a $p$-group.

A concrete realization of this data is given by the $U(1)^r$ Chern–Simons theory \cite{Belov:2005ze},
\be
\mathcal T: \quad 
S[A]=\frac{i}{4\pi} K_{IJ} \int_{\mathcal V} A_I \wedge dA_{J},\label{CSaction}
\ee
where $K$ is an integral, symmetric matrix of signature $(\lambda_+,\lambda_-)$ and with even diagonal entries. In other words, $K$ is the Gram matrix of a full-rank, even integral lattice $\Lambda$, $K=\Lambda^T g \Lambda$ with $g=\mathrm{diag}(1^{\lambda_+},-1^{\lambda_-})$.
The condition that the chiral central charge vanishes translates to
$\lambda_+-\lambda_- = 0 \pmod 8$. The $1$-form symmetry group is given by $\G=\mathbb Z^r/(K\mathbb Z^r)$ and the topological spins are $\theta(x)=\exp\bigl(\pi i\,x^T K^{-1} x\bigr)$.

We will only consider spaces $\mathcal V$ whose boundary is a connected surface $\Sigma_g$ of genus $g$.
For $\Sigma_g$, we choose a canonical basis of $1$-cycles $\gamma_1,\dots,\gamma_{2g}$ with intersection form
\be
\langle \gamma_i,\gamma_j\rangle=\eta_{ij}, 
\qquad 
\eta=\begin{pmatrix}
	0_{g} & -\id_{g}\\
	\id_{g} & 0_{g}
\end{pmatrix}.\label{intersection}
\ee
The homology group $H_1(\Sigma_g,\mathbb Z_N)\cong\mathbb Z_N^g\times\mathbb Z_N^g$ is thus a $\mathbb Z_N$-module equipped with the symplectic form $\eta$.

Anyon loops are described by
\be
\W^{\alpha}_x=\mathrm{Tr}_x\left[\exp\oint_\alpha A\right], \qquad \alpha\in H_1(\Sigma_g,\mathbb Z_N),\;~~ x\in\G.
\ee
Geometrically, these loops are knots (or more generally, links)\footnote{More precisely, $\alpha$ specifies an equivalence class of links modulo $N$. If $\alpha$ has order $N$, it lifts to a coprime integer tuple $\tilde \alpha\in H_1(\Sigma_g,\mathbb Z)$ describing an unknot. If instead $\alpha$ has order $d<N$, its minimal representative is a link of $d$ parallel unknots $\tilde \alpha/d$, where $d=\gcd(\alpha)$.} colored by elements of $\G$, embedded on the genus-$g$ surface $\Sigma_g$ without self-intersections \cite{Labastida:1990bt}. The links can be pushed to the interior of the space $\mathcal V$ bounding $\Sigma_g$, where they acquire framing.

\subsection{Anyon basis and Weyl operators\label{subsection:HW}}

We now review the construction of the anyon basis for the TQFT Hilbert space $\mathcal H$ on the surface $\Sigma_g$, as well as the Weyl basis for $\End(\mathcal H)$. The anyon basis plays the role of a computational basis for a system of $g$ qudits, each with local Hilbert space of dimension $|\G|$. While qudits are usually defined over $\mathbb Z_N$ or finite fields $\mathbb F_q$, more general constructions exist for arbitrary groups (see \cite{Bloomquist:2018hrp} in the context of anyons).

Consider first the torus Hilbert space $\mathcal H\cong \mathbb C^{|\G|}$. To construct the basis, we fill in the cycle $(0,1)\in H_1(\Sigma_1,\mathbb Z_N)$ of the torus to obtain a solid handlebody $\mathcal V_1$. The dual cycle $(1,0)$ supports non-contractible anyon lines. The path integral on $\mathcal V_1$ with such insertions defines the basis states
\be
|x\rangle\equiv\PI{\mathcal V_1,\W_x^{(1,0)}}=\hb{x}, 
\qquad x\in\G,
\ee
which we normalize by $\langle x|y\rangle=\delta_{x,y}$.
Thus $\{|x\rangle,{x\in\G}\}$ is an orthonormal basis for $\mathcal H$.

A natural basis for $\End(\mathcal H)$, the \textit{Weyl basis}, can be constructed by anyon insertions on the cylinder $\Sigma_1\times I$ with two boundaries of opposite orientation. 
 The path integral on $\Sigma_1\times I$ with line insertions constructs linear operators acting on $\mathcal H$. The matrix elements of such an operator can be computed by gluing handlebodies with anyon insertions to the two boundaries.

Inserting an anyon $\W^{\alpha}_x$ in $\Sigma_1\times I$ produces the corresponding anyon operator, which by slight abuse of notation we denote by the same symbol
\be
\W^{\alpha}_x\equiv\PI{\Sigma_1\times I,\W^{\alpha}_x}.
\ee
Choosing independent cycles in $H_1(\Sigma_1\times I,\mathbb Z_N)$, we first define
\be
X_x\equiv \W_x^{(1,0)}=\vcenter{\hbox{\begin{tikzpicture}[scale=0.7]
			\draw[thick] (0,0) to (0,2);
			\draw[thick] (0,2) to (2,2);
			\draw[thick] (2,2) to (2,0);
			\draw[thick] (2,0) to (0,0);
			\draw[blue] (0,1) to (2,1);
			\node[anchor = south] at (0.4,1) {$x$};
\end{tikzpicture}}}=\sum_{y,z} N_{xy}^z |z\rangle\langle y| 
=\sum_{y} |x+y\rangle\langle y|,
\ee
where we made use of the fusion rules $N_{xy}^z=\delta_{x+y,z}$.
Insertion along the dual cycle yields the operator
\be
Z_x\equiv \W_x^{(0,-1)}=\vcenter{\hbox{\begin{tikzpicture}[scale=0.7]
			\draw[thick] (0,0) to (0,2);
			\draw[thick] (0,2) to (2,2);
			\draw[thick] (2,2) to (2,0);
			\draw[thick] (2,0) to (0,0);
			\draw[blue] (1,0) to (1,2);
			\node[anchor = south] at (1.3,0) {$x$};
\end{tikzpicture}}}=\sum_y \omega^{\g(x,y)}|y\rangle\langle y|,
\ee
where the phase arises from the Hopf linking between $x$ and $y$.

The operators $X_x$ and $Z_x$ are generalizations of the Pauli operators familiar from the stabilizer formalism. A general anyon operator can be expressed as
\be
\W^{(a,b)}_x=\omega^{2^{-1}ab\,\q(x)}X_x^a Z_x^{-b}
= \omega^{-2^{-1}ab\,\q(x)} Z_x^{-b} X_x^a
\ee
and its action on the computational basis is \cite{Labastida_2001}
\be
\W^{(a,b)}_x|y\rangle
= \omega^{-2^{-1}ab\,\q(x)}\,
\omega^{-b\,\g(x,y)}\,|ax+y\rangle.
\ee
Two anyons obey the commutation relation
\be
\W_x^\alpha \W_y^{\alpha'}
= \omega^{\g(x,y)\eta(\alpha',\alpha)} \W_y^{\alpha'} \W_x^\alpha.
\ee
We see that two anyons commute either when they are inserted along orthogonal cycles with respect to $\eta$, or their charges are orthogonal with respect to $\g$. This is an early hint about the two types of abelian subgroups of anyons we will study in section \ref{sec2}.
There is redundancy in this notation, due to the relation
\be
\W^{\lambda\alpha}_x=\W^{\alpha}_{\lambda x}, 
\qquad \lambda\in\mathbb Z_N.
\ee
Collectively, the indices of anyon operators form the group
$\G\otimes \mathbb Z_N^2 \;\cong\;\G^2$, which is
a $\mathbb Z_N$-module equipped with the symplectic form $J=\g\otimes \eta$.
It is therefore useful to define the \textit{Weyl operators} $\W:\G^2\to \End(\mathcal H)$ as
\be
\W(a,b)\equiv\omega^{2^{-1}\g(a,b)}\,\W_a^{(1,0)}\W_b^{(0,-1)}=\weyl{a}{b},\quad (a,b)\in \G^2.
\ee

For abelian theories, this construction extends straightforwardly to genus-$g$, since the Hilbert space $\mathcal H \cong \mathbb C^{|\G|^g}$ factorizes into $g$ tensor copies of the torus Hilbert space, with basis states
\be
|x_1,\dots,x_g\rangle\equiv \hb{x_1}~\#\cdots\#~\hb{x_g}=|x_1\rangle\otimes\cdots\otimes|x_g\rangle,
\ee
where $\#$ above denotes the connected sum of the boundary surfaces.
Anyon operators decompose as
\be
\W^{(a_1,\dots,a_g,b_1,\dots,b_g)}_x
=\W_x^{(a_1,b_1)}\otimes\cdots\otimes\W_x^{(a_g,b_g)},
\ee
and similarly for the Weyl operators,
\be
\W(\alpha)=\W(\alpha_1)\otimes\cdots\otimes \W(\alpha_g),
\qquad \alpha=(\alpha_1,\dots,\alpha_g)\in (\G\times\G)^g.
\ee
The Weyl operators form the basis of $\End(\mathcal H)$ we have been looking for. This basis is orthonormal with respect to the Hilbert–Schmidt inner product
\be
\frac{1}{|\G|^g}\,\mathrm{tr}\left(\W(\alpha)^\dagger \W(\beta)\right)=\delta_{\alpha,\beta},
\qquad \alpha,\beta\in\G^{2g}.
\ee

The group $\G^{2g}$ can be viewed as the additive group of $n\times 2g$ matrices with columns in $\G$
\be
\alpha=\begin{pmatrix}
	a_1^{(1)} \cdots a_g^{(1)} & b_1^{(1)} \cdots b_g^{(1)}\\
	\vdots & \vdots \\
	a_1^{(n)} \cdots a_g^{(n)} & b_1^{(n)} \cdots b_g^{(n)}
\end{pmatrix}\in \G^{2g}.\label{matrixalpha}
\ee
The Weyl operators form a projective representation of $\G^{2g}$ with phases determined by the symplectic form $J=\g\otimes\eta$
\be
\W(\alpha)\W(\beta)=\omega^{-2^{-1}J(\alpha,\beta)}\W(\alpha+\beta).\label{prod1}
\ee
The group generated by these operators is the \textit{Weyl–Heisenberg group} (also sometimes called the generalized Pauli group on $g$ qudits),
\be
\mathcal P=\{\omega^k \W(\alpha)\mid \alpha\in\G^{2g},\;k\in\mathbb Z_N\},
\ee
which fits into the split short exact sequence
\be
\begin{tikzcd}
	0 \arrow[r] & \mathbb Z_N \arrow[r] & \mathcal P \arrow[r] & \G^{2g} \arrow[r] & 0
\end{tikzcd}.
\ee
Here $\mathbb Z_N\cong\{\omega^k\mid k\in\mathbb Z_N\}$ is the center of $\mathcal P$, and $\mathcal P/\mathbb Z_N\cong\G^{2g}$.

Since the Weyl operators form an orthonormal basis, any $\rho\in\End(\mathcal H)$ admits an expansion with coefficients
\be
\Xi_\rho(\alpha)=\frac{1}{|\G|^g}\,\mathrm{tr}\left(\W(\alpha)^\dagger \rho\right),
\qquad \alpha\in\G^{2g}.
\ee
The function $\Xi_\rho(\alpha)$ is called the \textit{characteristic function} of $\rho$.
A closely related quantity is the \textit{Wigner function} $\mathbb W_\rho(\alpha)$, defined as the symplectic Fourier transform of $\Xi_\rho$ \cite{Gross_2006}
\be
\mathbb W_\rho(\alpha)=\sum_{\beta} \omega^{-J(\alpha,\beta)} \Xi_\rho(\beta).
\ee
The Wigner function provides the expansion coefficients of $\rho$ in the Fourier-dual basis to the Weyl operators. This basis consists of ``dual anyon lines", which are lines confined to the invertible surface defect obtained by ``higher-gauging" $\G$ (see appendix \ref{highergauging}).

\subsection{Clifford group and its $\O(\G)\times \Sp(2g,\mathbb Z_N)$ subgroup}\label{sp2g}

The unitary operators on $\mathcal H$ also act on $\End(\mathcal H)$ by conjugation. The subgroup $\Aut(\mathcal P)$ of unitary automorphisms of the Weyl operators is called the \textit{Clifford group}. Its elements are unitaries $U$ such that
\be
U\W(\alpha)U^\dagger=\omega^{\phi}\W(\gamma(\alpha)),\ee
for some $\phi\in\mathbb Z_N$ and $\gamma:\G^{2g}\to\G^{2g}$.

Immediately from the definition, $\P\subseteq \Aut(\P)$, and in fact $\P$ is a normal subgroup. The Clifford group therefore fits into the short exact sequence (which splits for odd $N$ \cite{Tolar_2018})
\be
\begin{tikzcd}
	1 \arrow[r] & \mathcal P \arrow[r] & \Aut(\P) \arrow[r] & \Aut(\P)/\mathcal P \arrow[r] & 1
\end{tikzcd}.
\ee

Because Clifford elements act unitarily, they preserve the Weyl commutation relations, and hence $\gamma$ must preserve the symplectic form $J$. Denote by $\Sp(2g,\G)$ the group of transformations preserving $J$. For any $\gamma\in \Sp(2g,\G)$, there exists a unitary operator $U_\gamma$ such that
\be
U_\gamma \W(\alpha) U_\gamma^\dagger=\omega^{\lambda(\gamma,\alpha)}\W(\gamma(\alpha)).
\ee
The operators $\{U_\gamma\}$ form a projective representation of $\Sp(2g,\G)$, called the \textit{Weil representation}. Thus, the full Clifford group is generated by $\P$ together with $\{U_\gamma, \gamma\in \Sp(2g,\G)\}$. Equivalently, the Clifford group is a (projective) representation of the group of affine symplectic transformations \cite{Gross_2006}
\be
\alpha\mapsto \gamma(\alpha)+\alpha',\qquad \alpha,\alpha'\in\G^{2g},~\gamma\in \Sp(2g,\G),
\ee
where the above transformation is implemented by the Clifford operator $\W(\alpha')U_\gamma$.

In this work we are especially interested in the subgroup $\O(\G)\times \Sp(2g,\mathbb Z_N)\subseteq \Sp(2g,\G)$.
The orthogonal group $\O(\G)$ describes automorphisms of the quadratic module $(\G,\q)$, while the symplectic group $\Sp(2g,\mathbb Z_N)$ preserves the intersection form (\ref{intersection}) on $H_1(\Sigma_g,\mathbb Z_N)$.
Their actions are clearest in terms of the anyon operators.
For $Q\in \O(\G)$,
\be
U_Q \W_x^{\alpha}U_Q^\dagger=\W_{Q(x)}^{\alpha},
\ee
while for $\gamma\in \Sp(2g,\mathbb Z_N)$,
\be
U_\gamma \W_x^\alpha U_\gamma^\dagger=\W_x^{\gamma(\alpha)}.
\ee

In terms of Weyl operators, the combined action of $Q\in\O(\G)$ and $\gamma\in\Sp(2g,\mathbb Z_N)$ is
\be
\W(\alpha)\mapsto \W(Q\alpha\gamma),\qquad \alpha\in \G^{2g},
\ee
where $Q$ acts by left multiplication and $\gamma$ by right multiplication on the matrix form of $\alpha$ (\ref{matrixalpha}). Because any operator can be written as linear combination of Weyl operators, it follows that the actions of $\O(\G)$ and $\Sp(2g,\mathbb Z_N)$ commute.

The symplectic group $\Sp(2g,\mathbb Z_N)$ is a representation of the mapping class group of the genus-$g$ surface $\Sigma_g$. It is generated by
the local Hadamard $S$ and Phase $T$ gates, corresponding to torus modular transformations
\be
S=\frac{1}{\sqrt{|\G|}}\sum_{a,b\in\G} \omega^{-\g(a,b)}|a\rangle\langle b|,  
\qquad  
T=\sum_{a\in\G} \omega^{2^{-1}\q(a)}|a\rangle\langle a|,
\ee
along with the two-qudit CADD gate (generalizing the qubit CNOT), acting as
\be
\operatorname{CADD}:\quad|a\rangle\otimes |b\rangle\mapsto |a\rangle\otimes |a+b\rangle.
\ee
The $S$ and $T$ operators act independently on each torus in the connected sum $\Sigma_g\cong\Sigma_1\#\cdots\#\Sigma_1$, generating the subgroup $\operatorname{SL}(2,\mathbb Z_N)^g\subseteq \Sp(2g,\mathbb Z_N)$. Together with CADD (which couples pairs of tori), they generate the entire symplectic group.

\subsection{Stabilizer groups and stabilizer states}

An abelian subgroup $\mathcal S$ of the Weyl-Heisenberg group $\P$ is called a \textit{stabilizer group}. Since its elements commute, they can be simultaneously diagonalized and they fix an invariant subspace of $\mathcal H$, called the \textit{stabilizer code}.
From (\ref{prod1}), Weyl operators satisfy
\be
\W(\alpha)\W(\beta)=\omega^{-J(\alpha,\beta)}\W(\beta)\W(\alpha).
\ee
Thus, a stabilizer group corresponds to an \textit{isotropic submodule} $\mathcal C\subseteq \G^{2g}$ with respect to $J$, i.e. a submodule on which $J$ vanishes. Because $\P$ is a nontrivial central extension of $\G^{2g}$ by $\mathbb Z_N$, $\mathcal C$ alone does not uniquely determine $\mathcal S$, but we also need to make a choice of phases. For each isotropic $\mathcal C$ and vector $v\in\G^{2g}$, define the stabilizer group
\be
\mathcal S(\mathcal C,v)=\{\omega^{-J(v,\alpha)}\W(\alpha): \alpha\in\mathcal C\}.
\ee
Summing over $\mathcal S$ produces the projector onto the code subspace
\be
\Pi_{\mathcal C,v}=\frac{1}{|\mathcal C|}\sum_{\alpha\in\mathcal C}\omega^{-J(v,\alpha)}\W(\alpha).\label{sstate}
\ee

The centralizer of $\mathcal S$ in $\P$ is generally a non-abelian group, corresponding to the dual (with respect to $J$) module $\mathcal C^\perp$. A \textit{maximally isotropic} submodule satisfies $\mathcal C=\mathcal C^\perp$, and is also called \textit{Lagrangian}. In this case the projector reduces to a rank-one operator, i.e. a \textit{stabilizer state} \cite{Gross_2006}
\be
\Pi_{\mathcal C,v}=|v;\mathcal C\rangle\langle v;\mathcal C|.\label{pcv}
\ee
From (\ref{sstate}), we can immediately read its characteristic function 
\be
\Xi_{|v;\mathcal C\rangle\langle v;\mathcal C|}(\alpha)=\frac{1}{|\G|^g}\,\omega^{-J(v,\alpha)}\delta_{\alpha\in \mathcal C},
\ee
and hence its Wigner function is simply the indicator function
\be
\mathbb W_{|v;\mathcal C\rangle\langle v;\mathcal C|}(\alpha)=\frac{1}{|\G|^g}\,\delta_{\alpha\in \mathcal C+v}.
\ee

Up to a phase, the state defined by (\ref{pcv}) can be written as
\be
|v;\mathcal C\rangle=\W(v)|0;\mathcal C\rangle,
\ee
so for each fixed $\mathcal C$, the states $\{\W(v)|0;\mathcal C\rangle\}$ form an orthonormal basis of $\mathcal H$. One may view $|0;\mathcal C\rangle$ as the “ground state” defined by the module $\mathcal C$, with the rest of the states obtained by inserting anyon lines. In what follows, we make the choice $v=0$.

Two families of Lagrangian submodules are particularly important. Note that if $\mathcal C\subseteq \G$ is isotropic with respect to $\q$, then $\mathcal C\otimes \mathbb Z_N^{2g}\subseteq \G^{2g}$ is isotropic.
Similarly, if $\mathcal C'\subseteq \mathbb Z_N^{2g}$ is isotropic with respect to the symplectic form $\eta$, then $\G\otimes \mathcal C'\subseteq \G^{2g}$ is isotropic.
This motivates the following definition
\begin{definition}\label{def1}
	A submodule $\mathcal C\subseteq \G^{2g}$ is $\eta$-isotropic ($\eta$-Lagrangian) if $\mathcal C=\G\otimes C$ for some isotropic (Lagrangian) $C\subseteq \mathbb Z_N^{2g}$ with respect to $\eta$.
	
	A submodule $\mathcal C\subseteq \G^{2g}$ is $\q$-isotropic ($\q$-Lagrangian) if $\mathcal C=\mathcal L\otimes \mathbb Z_N^{2g}$ for some isotropic (Lagrangian) $\mathcal L\subseteq \G$ with respect to $\q$.
\end{definition}
These two types of Lagrangian submodules have a special meaning in the TQFT. The $\eta$-Lagrangian submodules define the first homology group of a 3-manifold with boundary $\Sigma_g$. The stabilizer state $|0;\mathcal C\rangle$ is prepared by the TQFT path-integral on this space, while $|v;\mathcal C\rangle$ are the orthonormal states obtained by line insertions.
A $\q$-Lagrangian submodule $\mathcal L\otimes\mathbb Z_N^{2g}$ describes topological boundary conditions (TBC) on the genus-$g$ surface. The associated stabilizer state can be obtained by gauging $\mathcal L\subseteq \G$ in the bulk of a genus-$g$ handlebody.

\subsection{Gauging 1-form symmetries in 3D\label{gaug}}

We will now review the gauging of $1$-form symmetries in 3D, which we will need in later sections. 

\subsubsection*{Gauging $1$-form symmetries in the bulk}

To gauge the group $\G$ inside the bulk of a handlebody $\mathcal V_g$, one inserts a network of anyon lines dual to a triangulation of $\mathcal V_g$ and sums over all charges in $\G$. If the result depends on the choice of triangulation, then the $1$-form symmetry group $\G$ is \textit{anomalous}. Different choices of quadratic forms $\q$ classify possible anomalies of the $1$-form symmetry \cite{Putrov:2025xmw}. In an abelian theory, a network of non-anomalous lines reduces to loops wrapping around each non-contractible cycle (see figure \ref{bulkgauge}).

$\G$ is non-anomalous if the anyon lines have trivial mutual braiding. Since by assumption $\q$ is non-degenerate, $\G$ itself cannot be gauged in the bulk. However, one can gauge a subgroup $\H\subseteq \G$ on which $\q$ vanishes. In this case, the sum is performed only over lines with charges in $\H$. The resulting theory, denoted $\mathcal T/\H$, contains lines in $\H^\perp$ with those in $\H$ identified \cite{Hsin_2019}. Thus the residual $1$-form symmetry group of $\mathcal T/\H$ is $\H^\perp/\H$. 

This gauging process is a projection onto the stabilizer code associated to $\H\otimes\mathbb Z_N^{2g}\cong H_1(\Sigma_g,\H)$, and the Hilbert space of $\mathcal T/\H$ can be viewed as the space of ``logical" states. In particular, if $\H$ is Lagrangian, then this gauging process constructs a stabilizer state, corresponding to the 1-dimensional Hilbert space of $\mathcal T/\H$.

\begin{figure}
	$$\sum_{x_1,\dots,x_g\in\mathcal \H}\hb{x_1}~\#\cdots\#~\hb{x_g}$$
	\caption{The network that performs the gauging of the symmetry $\H\subset\G$ in the handlebody $\mathcal V_g$ reduces to $g$ loops around the non-trivial cycles. It is necessary for the quadratic form $\q$ to vanish on $\H$.\label{bulkgauge}}
\end{figure}


\subsubsection*{Gauging $1$-form symmetries on codimension $1$ surface}

We can also gauge $\G$ along a co-dimension 1 surface to construct a surface operator. Concretely, consider inserting a genus-$g$ surface parallel to the two boundaries of $\Sigma_g\times I$ and construct a network of lines dual to a triangulation of this surface. Since the anyons in our theory have trivial crossings (the $F$-symbols can be chosen trivial), this process is independent of triangulation and hence gauging on this surface is always possible.

More generally, we can gauge any subgroup $\H\subseteq\G$ on the surface. The possible gaugings are classified by pairs $(\H,\nu)$, where $\nu\in H^2(\H,\mathbb C^\times)$ specifies a choice of discrete torsion. 
 Varying $\nu$ amounts to redefining trivalent junction phases while keeping $F$-symbols trivial \cite{Roumpedakis:2022aik}
\be \vcenter{\hbox{\begin{tikzpicture}[scale = 0.7]
			\draw[thick] (0,0) -- (2,0) -- (2,2) -- (0,2) -- (0,0);
			\draw  (1,0) -- (1,0.25) node[right] { $b$} -- (1,2);
			\draw (0,1) -- (0.5,1.2) node[above] {$a$} -- (1,1.4);
			\draw (1,0.6) -- (2,1);
			\draw [fill=black] (1,1.4) circle (0.04);
			\draw [fill=black](1,0.6) circle (0.04);
\end{tikzpicture}}}~\mapsto ~\nu(a,b)~\vcenter{\hbox{\begin{tikzpicture}[scale = 0.7]
			\draw[thick] (0,0) -- (2,0) -- (2,2) -- (0,2) -- (0,0);
			\draw  (1,0) -- (1,0.25) node[right] { $b$} -- (1,2);
			\draw (0,1) -- (0.5,1.2) node[above] {$a$} -- (1,1.4);
			\draw (1,0.6) -- (2,1);
			\draw [fill=black] (1,1.4) circle (0.04);
			\draw [fill=black](1,0.6) circle (0.04);
\end{tikzpicture}}}.\ee
Above we expressed a torus Weyl operator by resolving the intersection of the two torus cycles into a pair of trivalent junctions. Redefining the junction by a phase, we obtain an overall phase $\nu(a,b)=\omega^{2^{-1}B(a,b)}$ for some antisymmetric $B:\H\times\H\to\mathbb Z_N$. We introduced a coefficient $2^{-1}$ in the exponent for later convenience. 

To gauge $(\H,\nu)$, we sum over all charges in $\H$ (see figure \ref{surfgauge}).
The path integral on $\Sigma_g\times I$ with this network inserted produces a surface operator
\be \rho_{(\H,\nu)}={1\over |\H|^g}\sum_{\alpha\in H_1(\Sigma_g,\H)}\nu(\alpha)W(\alpha) .\label{Hnu}\ee
Gauging $(\H,\nu)$ generates a ``quantum symmetry" group $\G/\H\times \widehat \H$, which describes ``higher lines" supported on the surface. These operators form alternative bases for $\End(\mathcal H)$. Specifically, 

\begin{figure}
	$$\sum_{a_1,b_1,\dots,a_g,b_g\in\mathcal \H}\nu(a_1,b_1)~\weyl{a_1}{b_1}~\otimes\cdots\otimes~\nu(a_g,b_g)~\weyl{a_g}{b_g}$$
	\caption{Gauging of $\H\subseteq\G$ on a surface $\Sigma_g$, for some choice of discrete torsion $\nu\in H^2(\H,\mathbb C^\times)$.\label{surfgauge}}
\end{figure}

By construction, every surface operator is invariant under $\Sp(2g,\mathbb Z_N)$. It is well-known \cite{KS2} that surface operators correspond to Lagrangian submodules of $\G\times \bar\G$ with respect to the quadratic form $\q\oplus(-\q)$. They can thus be interpreted either as operators in the commutant of $\Sp(2g,\mathbb Z_N)$, or equivalently, as stabilizer states describing a TBC of the theory $\mathcal T\times \bar{\mathcal T}$. This fact will be very useful in constructing TBC states when $\mathcal T$ is itself a double, since classifying subgroups, choices of torsion and using equation (\ref{Hnu}) is easier than directly searching for all Lagrangian submodules. This will be the main topic of section \ref{sec3}.


\section{Lagrangian Submodules and their Stabilizer States}\label{sec2}

In this section we provide the full classification of the two types of Lagrangian submodules in definition \ref{def1}. In subsection \ref{etaLagr} we discuss the $\eta$-Lagrangian submodules, which describe the representative topologies that an abelian TQFT can distinguish. In subsection \ref{sec:classification} we discuss the $\q$-Lagrangian submodules, which describe topological boundary conditions. Finally, in subsection \ref{sec:CFT} we review the relation between topological boundary conditions and 2D CFTs.

\subsection{Classification of $\eta$-Lagrangian submodules and representative topologies}\label{etaLagr}

We begin with the $\eta$-Lagrangian submodules, which have the form $\G\otimes C$, where $C\subseteq \mathbb Z_N^{2g}$ is Lagrangian with respect to the intersection form $\eta$.  Since they contain a $\G$ tensor factor, both the submodules and their associated stabilizer states are invariant under $\O(\G)$. In fact, these modules are invariant under the entire group $\Aut(\G)$.

 The associated stabilizer states can be prepared by the TQFT path integral on some 3-manifold with boundary $\Sigma_g$, with the factor $C$ describing its first homology group. The TQFT path integral on \textit{any} space with boundary $\Sigma_g$ prepares one of these states \cite{Dymarsky:2024frx}. We describe the simplest representative topology corresponding to each of these states and later, in section \ref{sec:TQFTgrav}, we will define the gravitational path integral as a weighted sum over these representatives. 
 
 A complete classification of Lagrangian submodules of $\mathbb Z_N^{2g}$ is given by \cite{albouy2008matrixreductionlagrangiansubmodules} (Theorem 7), which we restate in our setting.
 \begin{theorem}[Classification of symplectic Lagrangian submodules]\label{theorem:sympl_codes}
 	Let $C$ be a submodule of $\mathbb Z_N^{2g}$ and let $N=\prod_{i=1}^k p_i^{s_i}$ be the prime factorization of $N$. Define the $g$-tuple $\vec d=(d_1,\dots,d_g)$ with
 	\begin{equation}\label{tuples}
 		d_1\mid d_2\mid\cdots\mid d_g\mid N, \quad 1\leq d_j\leq \prod_{i=1}^k p_i^{\lfloor s_i/2\rfloor},\quad j=1,\dots,g.
 	\end{equation}
 	Then $C$ is a Lagrangian submodule with respect to $\eta$ if and only if there exists $\gamma\in \Sp(2g,\mathbb Z)$ such that
 	\begin{equation}\label{gen}
 		\mathrm{diag}(N/d_1,N/d_2,\dots,N/d_g,d_1,d_2,\dots,d_g)\times\gamma
 	\end{equation}
 	is a generator matrix for $C$.
 \end{theorem}
 Thus, Lagrangian submodules that are isomorphic as groups belong to the same symplectic orbit. This does not hold for general isotropic submodules. While two isomorphic isotropic submodules $C\cong C'\subseteq \mathbb Z_N^{2g}$ are trivially isometric (since the symplectic form vanishes), their isometry does not always extend to an isometry of $\mathbb Z_N^{2g}$. The Lagrangian condition is therefore essential for Theorem \ref{theorem:sympl_codes} to apply, as emphasized in \cite{albouy2008matrixreductionlagrangiansubmodules}.
 
 Theorem \ref{theorem:sympl_codes} provides a classification of symplectic stabilizer states in terms of $\Sp(2g,\mathbb Z)$ orbits, generalizing the results of \cite{Feng:2024fzh}. To make the notation less cumbersome, define
 \be \mathbb Z_{\vec d}\equiv\mathbb Z_{d_1}\times\cdots\times\mathbb Z_{d_g},\qquad \mathbb Z_{N/\vec d}\equiv\mathbb Z_{N/d_1}\times\cdots\times\mathbb Z_{N/d_g},\ee
 so that the submodule generated by \eqref{gen} is isomorphic to $\mathbb Z_{\vec d}\times\mathbb Z_{N/\vec d}$.

We begin by describing the \textit{free} Lagrangian submodules $C\subseteq \mathbb Z_N^{2g}$, which have rank $g$ and admit a basis of $g$ elements. These submodules can be represented by regular handlebodies. To see this, note that a free submodule defines a state invariant under insertions of all anyons along the $g$ independent cycles of $H_1(\Sigma_g,\mathbb Z_N)$ determined by its basis. In other words, $C$ renders $g$ cycles of the genus-$g$ surface contractible, thereby defining a handlebody. The mapping class group $\Sp(2g,\mathbb Z)$ acts transitively on the free Lagrangian submodules, permuting handlebodies. Its representation on the Hilbert space has kernel $\Gamma(N)$ and its image is $\Sp(2g,\mathbb Z)/\Gamma(N)\cong \Sp(2g,\mathbb Z_N)$, which is precisely the subgroup of the Clifford group described in section \ref{sp2g}. The state $|0\rangle^{\otimes g}$ corresponds to the free submodule $C=\{(0,a)\mid a\in\mathbb Z_N^g\}$, given by the choice $\vec d=\vec 1$ and $\gamma=\id$ in \eqref{gen}. Varying $\gamma$, we obtain the states that belong to the $\Sp(2g,\mathbb Z_N)$-orbit of $|0\rangle^{\otimes g}$.  Specifically, the free submodule $C\gamma=\{(0,a)\gamma \mid a\in\mathbb Z_N^g\}$ for $\gamma\in \Sp(2g,\mathbb Z_N)$ yields the stabilizer state $U_\gamma |0\rangle^{\otimes g}$.
However, not all elements of $\Sp(2g,\mathbb Z_N)$ generate distinct stabilizer states by acting on $|0\rangle^{\otimes g}$. The subgroup that fixes this state is the congruence subgroup
\be \Gamma_0(N)=\left\{\begin{pmatrix}
	A & B\\
	0 & D
\end{pmatrix}\in \Sp(2g,\mathbb Z_N)\right\}.\ee
Hence the distinct free stabilizer states correspond to cosets in $\Sp(2g,\mathbb Z_N)/\Gamma_0(N)$. In \cite{Singal:2022dqu}, the following explicit count was obtained
\be \#\text{ free stabilizer states }=N^{g(g+1)\over 2}\prod_{i=1}^{k}\prod_{j=0}^{g-1}(p_i^{-(g-j)}+1).\ee

When $N$ is not squarefree, non-free Lagrangian submodules of $\mathbb Z_N^{2g}$ also exist, which correspond to non-handlebody topologies. The simplest representative topology can be described as a singular handlebody as follows. Consider gluing a torus cylinder $\Sigma_1^{(1)}\times [0,1]$ to a torus handlebody $\mathcal V_1^{(2)}$ along their boundaries $\Sigma_1=\Sigma_1^{(1)}\times\{0\}$ and $\Sigma_2=\partial\mathcal V_1^{(2)}$ via a covering map $f$. Writing the torus as $S^1\times S^1$ with angular coordinates $(\theta,\phi)$, define $f:\Sigma_1\to \Sigma_2$ by $f(e^{i\theta},e^{i\phi})=(e^{i{N\over d}\theta},e^{id\phi})$. The resulting glued manifold $\mathcal V_1^d=\mathcal V_{1}^{(2)}\cup_f (\Sigma_{1}^{(1)}\times [0,1])$ has first homology
\be H_1(\mathcal V_1^d,\mathbb Z_{N})\cong {\mathbb Z_{N}^2\over \langle (N/d,0),(0,d)\rangle}\cong \mathbb Z_{N/d}\times\mathbb Z_{d}.\ee
Contracting the gluing surface to a line, this topology can be thought as a handlebody with a non-invertible (unless $d=1$) defect line insertion along its non-contractible cycle. These are condensation defects \cite{Roumpedakis:2022aik} that can be constructed by gauging $(N/d)\G$ in a torus handlebody, or ``higher-gauging" the same group on a torus surface and then shrinking this surface to a line.

In general, a non-free Lagrangian submodule $C\subseteq\mathbb Z_N^{2g}$ defines a class of 3-manifolds with $H_1(M,\mathbb Z_N)\cong {\mathbb Z_N^{2g}}/C\cong C$. The corresponding stabilizer state is the state obtained by the TQFT path integral on such a manifold. These spaces can equivalently be represented by singular handlebodies with insertions of defect lines along their non-contractible cycles (see figure \ref{defectsfig}).
Each $\vec d\neq \vec 1$ specifies a space $\mathcal V_g^{\vec d}$ with $H_1(\mathcal V^{\vec d}_g,\mathbb Z_N)=\mathbb Z_{\vec d}\times\mathbb Z_{N/\vec d}$. The module (\ref{gen}) for $\gamma=\id$ is stabilized by the subgroup
\be \Gamma_{\vec d}=\left\{\begin{pmatrix}
	A & B\\
	C & D
\end{pmatrix}\in \Sp(2g,\mathbb Z_N)\bigg|\operatorname{diag}(d_1,\dots,d_g)C=\operatorname{diag}(N/d_1,\dots,N/d_g)B=0\right\} ,\ee
so the distinct topologies in the orbit correspond to the cosets $\Sp(2g,\mathbb Z_N)/\Gamma_{\vec d}$. Counting these cosets is significantly more involved than in the free case $\vec d=\vec 1$.

For $\gamma=\id$, since \eqref{gen} is diagonal, the corresponding state factorizes as
\begin{equation}\label{stabstate}
	|\Omega_{\vec d}\rangle=\bigotimes_{i=1}^g \frac{1}{\sqrt{|\G_{N/d_i}|}}\sum_{a\in \G_{N/d_i}} |a\rangle
\end{equation}
where $\G_{N/d_i}=(N/d_i)\G\subseteq \G$. We chose to normalize the states such that $\langle \Omega_{\vec d}|\Omega_{\vec d}\rangle=1$ and fixed the overall phase such that $\langle 0^g|\Omega_{\vec d}\rangle>0$. The rest of the states in the orbit are given by
\be |\Omega_{\vec d};\gamma\rangle\equiv U_\gamma|\Omega_{\vec d}\rangle,~~\gamma\in \Sp(2g,\mathbb Z_N)/\Gamma_{\vec d}.\ee
Thus, the tuple $\vec d$ defines the first homology group up to tensoring with $\mathbb Z_N$, while $\gamma\in \Sp(2g,\mathbb Z_N)/\Gamma_{\vec d}$ defines the inclusion of $\Sigma_g$ into the boundary of $\mathcal V_g^{\vec d}$.

\begin{figure}\label{defectsfig}
	$$|\Omega_{\vec d}\rangle=|\Omega_{d_1,\dots,d_g}\rangle=\defect{d_1}~\#\cdots\#~\defect{d_g}$$
	\caption{A representation of the topology $\mathcal V_g^{\vec d}$, corresponding to (\ref{gen}) with $\gamma=\id$, as a singular handlebody. The dashed line labeled $d_i$ is a condensation defect that can absorb and emit anyons with charges in $(N/d_i)\G$, modifying the first homology of the $i$-th factor to $\mathbb Z_{N/d_i}\times\mathbb Z_{d_i}$. The stabilizer state $|\Omega_{\vec d};\gamma\rangle$ can be obtained by acting with the mapping class group element $\gamma\in \Sp(2g,\mathbb Z)$ on the above configuration. }
\end{figure}

Finally, it is useful to collect all tuples $\vec d$ from Theorem \ref{theorem:sympl_codes} into a set $\mathbb T$. This set parametrizes the symplectic orbits and naturally forms a lattice defined as follows.
\begin{definition}\label{deflattice}
	Let $\mathbb T$ denote the set of tuples $\vec d$ satisfying (\ref{tuples}). For $\vec d,\vec d'\in\mathbb T$, with $\vec d=(d_1,\dots,d_g)$ and $\vec d'=(d_1',\dots,d_g')$, define the operations $\land$ and $\lor$ as follows
	\be \vec d\land \vec d'=(\operatorname{lcm}(d_1,d_1'),\cdots,\operatorname{lcm}(d_g,d_g')),\ee
	\be \vec d\lor \vec d'=(\gcd(d_1,d_1'),\cdots,\gcd(d_g,d_g')).\ee
\end{definition}
The cardinality of $\mathbb T$ (number of symplectic orbits) is given by the next proposition.
\begin{proposition}\label{orbitscount}
	Let $N$ and $k$ be as in theorem \ref{theorem:sympl_codes}. Then, $|\mathbb T|$ is equal to the number of distinct choices of $\vec d$. Specifically,
	\be |\mathbb T|=\prod_{i=1}^k \binom{\lfloor s_i/2\rfloor+g}{g}.\label{stabcount}\ee
\end{proposition}
\begin{proof}
	The exponents of $p_i$ in $\vec d$ form a non-decreasing sequence taking values $0,1,\dots,\lfloor s_i/2\rfloor$. The number of such sequences is $\binom{\lfloor s_i/2\rfloor+g}{g}$.
\end{proof}

\subsection{Classification of $\q$-Lagrangian submodules and topological boundary conditions}\label{sec:classification}

We now turn our attention to the $\q$-Lagrangian submodules, which are submodules of the form $\mathcal L\otimes \mathbb Z_N^{2g}$, where $\mathcal L\subset\G$ is Lagrangian with respect to $\q$. These submodules describe topological boundary conditions (TBC), which in turn describe 2D CFT partition functions on $\Sigma_g$. The factor $\mathcal L\subset\G$ is a classical even, self-dual code \cite{Barbar:2023ncl}, whose relation to Narain CFTs has been studied extensively in the recent literature \cite{Dymarsky:2020qom,Angelinos_thesis,Dymarsky:2020bps,Dymarsky:2020pzc,Yahagi:2022idq,Furuta:2022ykh,Henriksson:2022dnu,Angelinos:2022umf,Henriksson:2022dml,Dymarsky:2022kwb,Kawabata:2022jxt,Furuta:2023xwl,Alam:2023qac,Kawabata:2023iss,Ando:2024gcf,Barbar:2023ncl,Aharony:2023zit,Dymarsky:2024frx,Kawabata:2025hfd,AngelinosWZW,Dymarsky:2025agh} and will be very briefly reviewed in the next subsection. Since these submodules contain the full $\mathbb Z_N^{2g}$ factor, their stabilizer states are manifestly invariant under $\Sp(2g,\mathbb Z_N)$. In this subsection we will focus on groups of even length of the form $\G=\mathbb Z_N^{2n}$. We will provide a classification of these submodules and show that they are linearly independent as long as $n\leq g$. This discussion can be generalized to general finite abelian groups, but we will not attempt to treat the general case rigorously here.

Without loss of generality we can consider $N=p^m$. Since, by assumption, the quadratic form $\q$ is non-degenerate and $p$ is odd prime, there are exactly two inequivalent choices of $\q$, determined by whether $\det(\g)$ is a square or non-square in $\mathbb Z_N$  \cite{WALL1963281}, where $\g$ is the bilinear form associated with $\q$. We will focus on the case where $\g$ can be brought into the antidiagonal form
\be \g=\begin{pmatrix}
	0      & 0      & \cdots & 0      & 1 \\
	0      & 0      & \cdots & 1      & 0 \\
	\vdots & \vdots & \ddots& \vdots & \vdots \\
	0      & 1      & \cdots & 0      & 0 \\
	1     & 0      & \cdots & 0      & 0
\end{pmatrix}.\label{antidiag}\ee
Specifically, $\g$ lies in the square residue class, except when $n$ is odd and $p=3\mod 4$, in which case $\g$ belongs to the non-square residue class. In particular, the untwisted Dijkgraaf-Witten theory \cite{Dijkgraaf:1989pz} of gauge group $\mathbb Z_N^n$ belongs to this class.

Our first goal will be to show that the set of $\q$-Lagrangian stabilizer states $|\mathcal L_i\rangle$ forms a basis for the $\Sp(2g,\mathbb Z)$-invariant subspace of $\mathcal H$, as long as $n\leq g$. The central result of \cite{NRS} implies that the $\q$-Lagrangian stabilizer states span this subspace. To show linear independence, we use the notion of the \textit{minimal matrix}, which is an $n\times n$ matrix that uniquely specifies a Lagrangian submodule of $\G$ with respect to $\q$ (see appendix \ref{appB} for definition and proof).
\begin{proposition}\label{indep}
	If $g \geq n$, then the set of all $\q$-Lagrangian stabilizer states is linearly independent.
\end{proposition}
\begin{proof}
	Let $\{|\mathcal L_i\rangle\}$ denote the set of all $\q$-Lagrangian stabilizer states.  
	Suppose that for some $g \geq n$ we have
	$
	\sum_i c_i |\mathcal L_i\rangle = 0.
	$
	Let $M$ be the minimal matrix of $\mathcal L_j$ with rows $v_1,\dots,v_n$. Act with
	$
	\big(\langle v_1|\langle v_2|\cdots \langle v_n|\otimes \id^{\otimes (g-n)}\big)
	$
	on the relation above. Since $M$ uniquely determines $\mathcal L_j$, this annihilates all states except $|\mathcal L_j\rangle$, yielding
	$
	c_j=0
	$.
\end{proof}

The methods of \cite{albouy2008matrixreductionlagrangiansubmodules} can be straightforwardly adapted to the bilinear form (\ref{antidiag}) to imply the following classification.
\begin{theorem}[Classification of orthogonal Lagrangian submodules]
Let $\mathcal L$ be a submodule of $\mathbb Z_N^{2n}$ and let $N=\prod_{i=1}^k p_i^{s_i}$ be the prime factorization of $N$. Define an $n$-tuple $\vec q=(q_1,\dots,q_n)$ with
\begin{equation}\label{tuples2}
	q_1\mid q_2\mid\cdots\mid q_n\mid N, \quad 1\leq d_j\leq \prod_{i=1}^k p_i^{\lfloor s_i/2\rfloor},\quad j=1,\dots,n.
\end{equation}
	Then $\mathcal L$ is a $\q$-Lagrangian submodule if and only if there exists $Q\in\O(\G)$ such that
	\be \operatorname{diag}(q_1,\dots,q_{n},N/q_{n},\dots,N/q_1)\times Q,\label{eseven}\ee
	is a generator matrix for $\mathcal L$.	
	\label{theorem:Orthogonal_codes}
\end{theorem}
The classification of orbits for an equivalent bilinear form $\g'$ on $\G$ can be obtained by acting on (\ref{eseven}) with the automorphism of $\G$ that brings \eqref{antidiag} to $\g'$.

As in the symplectic case, the Lagrangian property is vital to the classification (\ref{eseven}), as two isomorphic isotropic submodules are not necessarily related by a global isometry.

Since (\ref{eseven}) for $Q=\id$ is diagonal, its stabilizer state is a factorized state given by
\be |\mathcal L_{\vec q,\id}\rangle=\left(\sum_{(a_i,b_i)\in\mathbb Z_{N/q_i}\times\mathbb Z_{q_i}} |a_1q_1,\dots,a_nq_n,b_1N/q_n,\dots,b_n N/q_1\rangle\right)^{\otimes g}.\label{qstate}\ee
We shall also denote these states as $|\mathcal L_{\vec q}\rangle$ by dropping the $\id$ in the subscript.
We can then obtain the general stabilizer state of (\ref{eseven}) by
\be |\mathcal L_{\vec q,Q}\rangle= U_{Q}|\mathcal L_{\vec q}\rangle,\quad Q\in\O(\G).\ee
We normalize the $\q$-Lagrangian stabilizer states such that $\langle 0^g|\mathcal L_{\vec q,Q}\rangle=1$, where $0^g$ denotes the identity of the group $\G^g$. With this normalization, they map straightforwardly to a CFT partition function.

We now define the set $\mathbb S$ containing all tuples in (\ref{tuples2}) classifying the orthogonal orbits. Analogously to $\mathbb T$, this set has a lattice structure.
The counting of orthogonal stabilizer states for $\G=\mathbb Z_N^{2n}$ is given by a formula analogous to (\ref{stabcount}) (using the notation of theorem \ref{theorem:Orthogonal_codes})
\be |\mathbb S|=\prod_{i=1}^k \binom{\lfloor s_i/2\rfloor+n}{n}.\label{stabcount2}\ee

\subsection{CFT partition functions from $\q$-Lagrangian submodules}\label{sec:CFT}

We now briefly review the 2D conformal field theories (CFTs) that are dual to the 3D $U(1)^r$ CS theories described by (\ref{CSaction}).  
These theories are bosonic lattice CFTs, with each TBC of the CS theory specifying a modular-invariant partition function of the corresponding CFT \cite{Barbar:2023ncl}. The CFT partition function can be obtained by the sandwich construction \cite{Freed:2022qnc} as follows. Consider $\mathcal T$ on the cylinder $\Sigma_g\times [0,1]$ with TBC specified by $\mathcal L_{\vec q,Q}$ imposed at $\Sigma\times\{0\}$ and conformal boundary condition at $\Sigma_g\times\{1\}$. The states on the two boundaries are respectively $|\mathcal L_{\vec q,Q}\rangle$ and $\langle \Omega|$, where by $\Omega$ we denote collectively the parameters specifying the conformal boundary condition. Since the bulk theory is topological, the interval can be collapsed, resulting in the partition function
\be Z_{\vec q,Q}(\Omega)=\langle \Omega|\mathcal L_{\vec q,Q}\rangle.\ee

As discussed after (\ref{CSaction}), the $K$-matrix in the CS action is the Gram matrix of an even integral lattice $\Lambda$.  The anyon basis $\{|a\rangle,a\in\G\}$ of the TQFT maps to the ``non-holomorphic blocks'' $\langle \Omega|a\rangle$ of the CFT \cite{Belov:2005ze}, corresponding to the cosets $\Lambda^\perp / \Lambda\cong \G$.  
Given a TBC specified by $\mathcal L_{\vec q,Q}$, one constructs the associated Narain lattice $\Lambda_{\mathcal L_{\vec q,Q}}$ via \emph{construction~A} \cite{SPLAG} applied to the pair ($\mathcal L_{\vec q,Q}$, $\Lambda$) \cite{Angelinos:2022umf,Aharony:2023zit}. 
The CFT partition function $Z_{\vec q,Q}$ is given by the  genus-$g$  Siegel-Narain theta series of $\Lambda_{\mathcal L_{\vec q,Q}}$, normalized appropriately (explicit formulas may be found in \cite{Henriksson:2022dnu,Benini:2022hzx}).  
When the chiral central charge satisfies $c = \pm 8 \pmod{24}$, modular invariance requires an additional factor of the $(E_8)_1$ partition function in the left- or right-moving sector. Thus, for a fixed $\Lambda$, each Lagrangian $\mathcal L_{\vec q,Q}\subseteq \G$ gives rise to a CFT partition function $Z_{\vec q,Q}$.

\section{The Quantum Gravity Path-Integral\label{sec:TQFTgrav}}

The quantum gravity path integral for a topological theory involves a sum over topologies. According to \cite{Dymarsky:2024frx}, the path integral of an abelian TQFT $\mathcal T$ on any topology with $\Sigma_g$ boundary is a stabilizer state described by an $\eta$-Lagrangian submodule. 
For a fixed $\mathcal T$, this partitions all the possible topologies into a finite number of equivalence classes, with the path integral of $\mathcal T$ preparing the same state within each class. We can thus reframe the gravitational path integral as a weighted sum over the representative topologies defined in subsection \ref{etaLagr}.

As discused in subsection \ref{etaLagr}, the representative topologies are parametrized by $\vec d\in\mathbb T$ and $\gamma\in \Sp(2g,\mathbb Z_N)/\Gamma_{\vec d}$. The $g$-tuple $\vec d$ describes the first homology group of the space, while $\gamma$ specifies an inclusion map of the surface $\Sigma_g$ into $\partial\mathcal V_g^{\vec d}$. Since we require consistency with holography, it is reasonable to assume that for each fixed $\vec d$, the sum should include a uniform average over the symplectic orbit, so that the resulting state is invariant under the mapping class group of $\Sigma_g$. With this assumption, the resulting state can be expanded in terms of TBC states, which are $\q$-Lagrangian stabilizer states. The latter expression can be interpreted as a weighted average over 2D CFT partition functions \cite{Barbar:2023ncl}.

If $N$ is squarefree, this construction is particularly simple. All $\eta$-Lagrangian stabilizers belong to a single symplectic orbit and they describe regular handlebody topologies $\mathcal V_g$ with $H_1(\mathcal V_g,\mathbb Z_N)=\mathbb Z_N^g$. The uniform average over these regular handlebodies gives a well-defined ensemble of 2D CFTs valid at any genus \cite{Dymarsky:2025agh}.

If $N$ is not square-free, there also exist $\eta$-Lagrangian submodules that describe non-handlebody topologies. Summing only over the handlebodies leads to an ill-defined ensemble of CFT partition functions. This issue is rectified by averaging over \textit{all} representative topologies, with an appropriate measure. As emphasized earlier, within each symplectic orbit we average uniformly, so that the resulting state is $\Sp(2g,\mathbb Z_N)$-invariant. We define the TQFT gravity path integral as a weighted average over all representative topologies of boundary $\Sigma_g$ defined by the $\eta$-Lagrangian stabilizer states
\be \mathcal Z_{gravity}[\mathcal T,\Sigma_g;w]\equiv\sum_{\vec d\in\mathbb T}w_{\vec d}\avg{|\Omega_{\vec d}\rangle}.\label{gravityZdef}\ee
Above $\avg{\cdot}$ denotes a uniform average over the symplectic orbit, which will be defined precisely in subsection \ref{topavg}.
We call $w_{\vec d}$ the \textit{bulk weights}.

In subsection \ref{topavg} we express the path integral (\ref{gravityZdef}) as follows
\be \mathcal Z_{gravity}[\mathcal T,\Sigma_g;w]=\sum_{\vec q\in\mathbb S}m_{\vec q}|\overline{\mathcal L_{\vec q}}\rangle,\label{ZZg}\ee
where $\vec q\in\mathbb S$ parametrizes the $\O(\G)$-orbits of TBCs and $|\overline{\mathcal L_{\vec q}}\rangle$ denote the uniform average of TBCs within an orbit. We call the coefficients $m_{\vec q}$ the \textit{boundary weights}. 

In subsection \ref{DSpresc} we review the genus-reduction prescription of \cite{Dymarsky:2024frx}, which fixes the boundary weights and consequently also fixes\footnote{The bulk weights are not fixed uniquely when $|\mathbb T|<|\mathbb S|$, since the $\eta$-Lagrangian stabilizer states are linearly dependent.} the bulk weights $w_{\vec d}$.

\subsection{Average over topologies}\label{topavg}

In this subsection we derive the map between the bulk and boundary weights. We will assume that $\G=\mathbb Z_N^{2n}$, but this calculation can be applied with minor modifications to general finite abelian group. This method was first used in \cite{Castro:2011zq} to calculate the Poincar\'e series of the vacuum in some Virasoro minimal models on a torus. It was also used in \cite{Raeymaekers:2021ypf,Meruliya:2021utr} for different RCFTs, and more recently in \cite{Nicosanti:2025xwu} for abelian TQFTs with general orientable boundaries.

We start by defining the mapping class group average of a seed state in the Hilbert space of a genus-$g$ surface $|\rho_{seed}\rangle\in\mathcal H$ as follows
\be \avg{|\rho_{seed}\rangle}\equiv {|\G|^g\over |Sp(2g,\mathbb Z_{N})|}\sum_{\gamma\in Sp(2g,\mathbb Z_{N})}U_\gamma |\rho_{seed}\rangle.\label{defrhoavg}\ee
By construction, this expression is $\Sp(2g,\mathbb Z)$-invariant and it can be expanded in terms of TBC states, according to subsection \ref{sec:classification}. Our goal will be to calculate the expansion coefficients
\be \avg{|\rho_{seed}\rangle}=\sum_i c_i|\mathcal L_i\rangle.\ee
Taking inner product with $\langle \mathcal L_i|$ on both sides of the equation we obtain
\be \langle \mathcal L_j|\rho_{\text{seed}}\rangle=\sum_{i=0}^m c_i D_{ij},~~~D_{ij}\equiv {1\over |\G|^g}|\mathcal L_i\cap \mathcal L_j|^g ,\label{eqscD}\ee
where we used the fact that $|\mathcal L_j\rangle$ is $\Sp(2g,\mathbb Z_N)$-invariant to simplify the LHS and we used $\langle\mathcal L_j| \mathcal L_i\rangle=|\mathcal L_i\cap \mathcal L_j|^g$ on the RHS, where $\mathcal L_i\subset\G$ is the Lagrangian submodule corresponding to stabilizer state $|\mathcal L_i\rangle$. The matrix $D$ defined above is the \textit{intersection matrix}.

If $g\geq n$, the states  $|\mathcal L_i\rangle$ are linearly independent (by proposition \ref{indep}), therefore $D$ is invertible. For $g<n$, $D$ may fail to be invertible. This simply means that the coefficients $c_i$ are not uniquely determined. The obvious way to fix this issue is to expand the state in a linearly independent subset of $|\mathcal L_i\rangle$, however this approach is not practical. Instead, we are going to treat $g$ as a positive real parameter, so that $D$ is invertible except on a discrete subset of values and $c_i$ can be calculated as functions of $g$. Then, $g$ can be sent to the desired value by taking a limit. With this in mind, we can formally treat $|\mathcal L_i\rangle$ as a basis for all $g$ and solve for $c$ in (\ref{eqscD}) 
\be c=D^{-1}v^{\text{seed}},~~v_j^{\text{seed}}\equiv\langle{\mathcal L_j}|\rho_{\text{seed}}\rangle\label{c-solns}.\ee
Therefore, we can write
\be \avg{ |\rho_{\text{seed}}\rangle}=\sum_{i,j}(D^{-1})_{ij}v_i^{\text{seed}}|\mathcal L_j\rangle \label{avgexpr}.\ee

\subsubsection*{The sum over handlebody topologies}
As a warmup, consider the seed state $|\rho_{\text{seed}}\rangle=|\Omega_{\vec 1}\rangle=|0^g\rangle$. The symplectic average of this state is the sum over regular handlebody topologies, in other words, it is the Poincar\'e series of the vacuum. Since each module $\mathcal L_j$ contains the zero element exactly once, it is clear that
\be v_j^{\vec 1}=\langle\mathcal L_j|0^g\rangle={1}\ee
and inserting into (\ref{avgexpr}) we find the average
\be  \avg{ |0^g\rangle}=\sum_{i,j}(D^{-1})_{ij}|\mathcal L_i\rangle\label{vacuumseries}.\ee
This expresses the sum over handlebodies in terms of the intersection matrix of the $\q$-Lagrangian submodules.

\subsubsection*{The sum over non-handlebody topologies}
Consider now $|\rho_\text{seed}\rangle=|\Omega_{\vec d}\rangle$, for a general $\eta$-Lagrangian stabilizer state (\ref{stabstate}). First, we calculate the seed vector $v^\text{seed}\equiv v^{\vec d}$
\be v_i^{\vec d}\equiv\langle \mathcal L_i|\Omega_{\vec d}\rangle\label{overlap} .\ee
For such overlaps there is a significant simplification. Since the state $|\Omega_{\vec d}\rangle$ is invariant under the entire $\Aut(\G)$ group, we can choose an automorphism of $\G$ that makes the bilinear form anti-diagonal and then bring the Lagrangians $\mathcal L_i$ to the diagonal form (\ref{eseven}). Then, $v_i^{\vec d}$ is constant on each $\O(\G)$-orbit.

To proceed it is convenient to change the indices in (\ref{avgexpr}) to the pairs $(\vec q, Q)$, where as in subsection \ref{sec:classification}, $\vec q\in\mathbb S$ denotes an $\O(\G)$-orbit of TBC states and $Q\in \O_{\vec q}\equiv\O(\G)/\Gamma_{\vec q}$ labels a TBC state within the orbit. Here $\Gamma_{\vec q}$ is the subgroup of $\O(\G)$ that stabilizes the diagonal Lagrangian in (\ref{eseven}). The intersection matrix with the new labeling reads
\be D_{\vec q,Q;\vec q',Q'}={|(\mathcal L_{\vec q,Q})\cap (\mathcal L_{\vec q',Q'})|^g\over |\G|^g},\ee
where $\mathcal L_{\vec q,Q}$ denotes the module in (\ref{eseven}).

We can express the overlap (\ref{overlap}) as
\be v_{\vec q,Q}^{\vec d}=\langle\mathcal L_{\vec q,Q}|\Omega_{\vec d}\rangle=\langle\mathcal L_{\vec q}|U_Q^\dagger|\Omega_{\vec d}\rangle=\langle\mathcal L_{\vec q}|\Omega_{\vec d}\rangle,\ee
where we used the fact that $|\Omega_{\vec d}\rangle$ is $\O(\G)$-invariant. Since $v$ does not depend on $Q$, we can drop this index.
Using the explicit expressions (\ref{qstate}) and (\ref{stabstate})
\be |\mathcal L_{\vec q}\rangle=\left(\sum |a_1q_1,\dots,a_nq_n,b_1N/q_n,\dots,b_n N/q_1\rangle\right)^{\otimes g} ,\ee
\be |\Omega_{\vec d}\rangle=\prod_{j=1}^g \bigotimes_{i=1}^n\left({1\over \sqrt{d_i}}\sum_{a\in\mathbb Z_{d_i}} |Na/d_i\rangle\right)^{\otimes 2} ,\ee
the overlap is given by
\be \langle \mathcal L_{\vec q}|\Omega_{\vec d}\rangle=\prod_{i=1}^{n}\prod_{j=1}^g \mathsf I(d_j,q_i),\ee
where
\be \mathsf I(d_j,q_i)\equiv {1\over d_j}\left|\langle N/d_j\rangle\cap\langle q_i \rangle\right|\cdot\left|\langle N/d_j\rangle\cap\langle N/q_i \rangle\right|={\gcd(d_j,q_i)\gcd(d_j,N/q_i)\over d_j} .\ee
We can simplify this expression to
\be \mathsf I(d_j,q_i)=\gcd(d_j,q_i,N/d_j,N/q_i)=\gcd(d_j,q_i),\ee
where in the last step we used that $q_i^2|N$ and $d_j^2|N$, according to their definitions. Putting everything together, we find the seed vector
\be v_{\vec q}^{\vec d}=\prod_{i=1}^{n}\prod_{j=1}^g \gcd(d_j,q_i) .\label{vmatr}\ee
Plugging into (\ref{avgexpr}) we have
\be \avg{ |\Omega_{\vec d}\rangle}= \sum_{\vec q,\vec q'\in\mathbb S}\sum_{Q\in \O_{\vec q}}\sum_{Q'\in \O_{\vec q'}}v_{\vec q}^{\vec d}D^{-1}_{\vec q,Q;\vec q',Q'}|\mathcal L_{\vec q',Q'}\rangle=\sum_{\vec q,\vec q'\in\mathbb S}\sum_{Q'\in \O_{\vec q'}}v_{\vec q}^{\vec d}Y_{\vec q;\vec q',Q'}{1\over |\O_{\vec q'}|}|\mathcal L_{\vec q',Q'}\rangle ,\label{bb}\ee
where we defined the matrix
\be Y_{\vec q;\vec q',Q'}\equiv |\O_{\vec q'}|\sum_{Q\in\O_{\vec q}} D^{-1}_{\vec q,Q;\vec q',Q'}.\ee
Since the LHS of (\ref{bb}) is invariant under $\O(\G)$, the RHS must also be invariant. This implies that $Y_{\vec q;\vec q',Q'}$ does not depend on the index $Q'$ so we can drop it. We defined $Y_{\vec q;\vec q'}$ this way so that it is symmetric. We define the  $\O(\G)$-averaged state
\be |\overline {\mathcal L_{\vec q}}\rangle\equiv{1\over |\O_{\vec q}|}\sum_{Q\in\O_{\vec q}}|\mathcal L_{\vec q,Q}\rangle, \label{Oavg}\ee
and we can finally write the expression
\be \avg{|\Omega_{\vec d}\rangle}=\sum_{\vec q,\vec q'\in\mathbb S}v_{\vec q}^{\vec d}Y_{\vec q;\vec q'}|\overline {\mathcal L_{\vec q'}}\rangle=\sum_{\vec q'\in\mathbb S}\t_{\vec d,\vec q'}|\overline {\mathcal L_{\vec q'}}\rangle.\label{bbf}\ee
Above we introduced the ``$\t$-matrix"
\be \t_{\vec d,\vec q}\equiv\sum_{\vec q'\in\mathbb S} v_{\vec q'}^{\vec d}Y_{\vec q';\vec q} \label{tau}.\ee
The TQFT gravity path integral can now be expressed as 
\be \mathcal Z_{gravity}[\mathcal T,\Sigma_g;w]=\sum_{\vec d\in\mathbb T}w_{\vec d}\avg{|\Omega_{\vec d}\rangle}=\sum_{\vec d\in\mathbb T,\vec q\in\mathbb S}w_{\vec d}\t_{\vec d,\vec q}|\overline {\mathcal L_{\vec q}}\rangle=\sum_{\vec q\in\mathbb S}m_{\vec q}|\overline {\mathcal L_{\vec q}}\rangle.\label{gravityZ}\ee
Above we defined the boundary weights
\be m_{\vec q}\equiv \sum_{\vec d\in\mathbb T} w_{\vec d}\t_{\vec d,\vec q}.\label{wtom}\ee
Thus the $\t$-matrix in (\ref{tau}) is an explicit linear map between the bulk and boundary weights, which can be completely determined from the knowledge of all TBCs that the TQFT admits. Classifying the full set of TBCs is not an easy problem and a systematic procedure based on the groupoid structure of \cite{Gaiotto_2021} will be discussed in section \ref{sec3}.

Some comments are in order. In the special case where $N$ is square-free, both $\mathbb S,\mathbb T$ have cardinality $1$. In other words, there is a single $\O(\G)$-orbit of TBC and a single $\Sp(2g,\mathbb Z_N)$-orbit of representatives topologies. Thus, there is a single state invariant under both groups and (\ref{gravityZ}) becomes
\be \mathcal Z_{gravity}[\mathcal T,\Sigma_g;w]=w_{\vec 1}\avg{ |\Omega_{\vec 1}\rangle}=w_{\vec 1}\t_{\vec 1,\vec 1}|\overline{\mathcal L_{\vec 1}}\rangle={w_{\vec 1}\t_{\vec 1;\vec 1}\over |\O_{\vec 1}|}\sum_{Q'\in\O_{\vec 1}}|\mathcal L_{\vec 1,Q'}\rangle ,\ee
which was proven recently in \cite{Dymarsky:2025agh}.

More generally, the TBC states $|\mathcal L_{\vec q,Q}\rangle$ span the $Sp(2g,\mathbb Z_N)$-invariant subspace and by averaging them over $\O(\G)$ we obtain the states $\{|\overline{\mathcal L_{\vec q}}\rangle\}$ which span the subspace $\mathcal K\subseteq\mathcal H$ invariant under $Sp(2g,\mathbb Z_N)\times \O(\G)$. Similarly, the states $|\Omega_{\vec d}\rangle$ span the $\O(\G)$-invariant subspace \cite{NRS} and their symplectic averages $\{\avg{|\Omega_{\vec d}\rangle},\vec d\in\mathbb T\}$ form another spanning set of $\mathcal K$. The dimension of $\mathcal K$ is given by $\dim\mathcal K=\min\{|\mathbb S|,|\mathbb T|\}$.

In equation (\ref{bbf}) a convex sum of $\eta$-Lagrangian stabilizer states is expressed as linear combination of $\q$-Lagrangian stabilizer states. Positivity of the Wigner function of projectors onto stabilizer states implies that the entries $\t_{\vec d,\vec q}$ are nonnegative, as long as $g\geq n$. For $g<n$, $\t_{\vec d,\vec q}$ are not uniquely defined, but can always be brought to a nonnegative form.
Therefore, if we require the weights $w_{\vec d}$ to be nonnegative, there can be no cancellations on the RHS of (\ref{gravityZ}) and we unavoidably obtain a (weighted) ensemble of all TBC states. The relation between the weights (\ref{wtom}) is generally not one-to-one, since $\t$ is not a square matrix, unless $|\mathbb S|=|\mathbb T|$.

For $g=n$, the matrix $\t_{\vec d,\vec q}$ is square and is in fact ``upper triangular" in the sense that for $g= n$, we have $\t_{\vec d,\vec q}=0$ unless $\vec d\leq\vec q$. To see this, note that constructing a TBC state $|\mathcal L_{\vec q,Q}\rangle$ involves inserting a network of anyon lines along all cycles of $\mathcal V_g^{\vec d}$. But when $\vec d\not\leq \vec q$, the space $\mathcal V_g^{\vec d}$ cannot support the network that constructs $|\mathcal L_{\vec q}\rangle$.

\subsection{A prescription for the weights}\label{DSpresc}

So far, we have not imposed any constraints on the weights $w_{\vec d}$. In a consistent definition of the quantum gravity path integral, however, the bulk weights should be completely fixed. A concrete proposal was given in \cite{Dymarsky:2024frx}, where the quantum gravity state was defined by applying a genus-reduction map to a uniform average over large-genus handlebodies. We now briefly review this prescription.

Let $\mathcal H_g$ denote the Hilbert space the TQFT assigns to the surface $\Sigma_g$.
 The dual vector $\langle 0|^{\otimes g'}\in\mathcal H_{g'}^\star$ can be viewed as a genus-reduction map $\langle 0|^{\otimes g'}:\mathcal H_{g+g'}\to\mathcal H_{g}$. Applying this map to (\ref{vacuumseries}) we obtain
\be  \begin{split}\langle 0|^{\otimes g'}\avg{|0\rangle^{\otimes (g+g')}}&=\sum_{i,j}(D_{g+g'}^{-1})_{i,j}(\langle 0|^{\otimes g'})|\mathcal L_i^{(g+g')}\rangle\\&=\sum_{i,j}(D_{g+g'}^{-1})_{i,j}|\mathcal L_i^{(g)}\rangle, \end{split}\label{vacuumseries2}\ee
where we introduced indices to explicitly denote the genera at which the states and the intersection matrix are evaluated.
Taking the limit $g'\to\infty$, the intersection matrix simplifies $D=\id$ (see (\ref{eqscD})) and we can formally write
\be \lim\limits_{g'\to\infty}\langle 0|^{\otimes g'}\avg{|0\rangle^{\otimes (g+g')}}=\sum_{i}|\mathcal L_i^{(g)}\rangle.\label{vacuumseriesinfty}\ee
The LHS of the above equality can be expressed as a linear combination of $\eta$-Lagrangian stabilizer states in which all representative topologies $|\Omega_{\vec d};\gamma\rangle$ inevitably arise. For a more rigorous discussion and an illustrative example we refer to \cite{Dymarsky:2024frx}. The boundary weights in the RHS of (\ref{vacuumseriesinfty}) are genus-independent, compatible with a holographic interpretation.

\section{Constructing Topological Boundary Conditions}\label{sec3}

In this section we describe a systematic procedure, based on \cite{Gaiotto_2021}, for constructing topological boundary conditions (TBCs) in a class of abelian TQFTs.
An abelian TQFT is specified by data $(\G,\q)$, where $\G$ is a finite abelian group of odd exponent $N$ and $\q$ is a quadratic form. Different quadratic forms related by automorphisms of $\G$ define equivalent TQFTs. The intersections of $\q$-Lagrangian modules are invariant under automorphisms of $\G$, therefore the $\t$-matrix of section \ref{sec:TQFTgrav} is the same for all equivalent theories. For concrete calculations it is convenient to fix a representative quadratic form, and for an explicit construction of the CFT ensemble (see section \ref{sec:CFT}) one also needs a concrete Chern–Simons realization \eqref{CSaction}. 
We will therefore adopt specific conventions.

In \ref{class1} we consider TQFTs equivalent to untwisted Dijkgraaf-Witten theories \cite{Dijkgraaf:1989pz} with gauge group $\sG$ of odd exponent. In these theories, we have a decomposition $\G=\sG\times\bar\sG$ and the quadratic form can be brought into the block-diagonal form $\q=q\oplus(-q)$ by an automorphism of $\G$, where $q$ is a quadratic form on $\sG$. The TQFT thus decomposes into sectors $\mathcal T \times \bar{\mathcal T}$, where $\mathcal T=(\sG,q)$ is described by a Chern–Simons theory (\ref{CSaction}) with positive-definite $K$-matrix. The pair $(\sG,q)$ defines a pointed MTC, which encodes the chiral data of a bosonic RCFT \cite{Fuchs:2002cm}. Elements of the anyon basis correspond to chiral conformal blocks \cite{Witten:1988hf}. We will therefore refer to $\mathcal T=(\sG,q)$ as the \textit{chiral} sector of the full theory $\mathcal T\times\bar{\mathcal T}$.
The Hilbert space of the doubled theory $\mathcal T \times \bar{\mathcal T}$ on $\Sigma_g$ is isomorphic to $\End(\mathcal H_{\mathcal T})$, where $\mathcal H_{\mathcal T}$ is the Hilbert space of $\mathcal T$ on $\Sigma_g$. The isomorphism is given by ``vectorization"
\be \End(\mathcal H_{\mathcal T})\to\mathcal H_{\mathcal T\times\bar{\mathcal T}}:~~|a\rangle\langle b|\mapsto |a,b\rangle,\quad a,b\in\sG^g.\label{vectorization}\ee
The inverse vectorization map transforms TBC states of $\mathcal T\times\bar{\mathcal T}$ into surface operators of $\mathcal T$. On the other hand, the representative topologies determined by $\eta$-Lagrangian submodules are not affected by this map. Each $\eta$-Lagrangian stabilizer state of $\mathcal H_{\mathcal T\times\bar{\mathcal T}}$ maps to a projector onto the corresponding stabilizer state of $\mathcal H_{\mathcal T}$.

More generally, in \ref{qiso}, we consider theories that can be brought to the form $\mathcal T \times \bar{\mathcal T}'$ such that $\mathcal T'$ is obtained from $\mathcal T$ by gauging an isotropic subgroup $\H\subset\sG$, i.e. $\mathcal T' = \mathcal T/\H$. The two chiral sectors now describe different RCFTs and this construction guarantees their compatibility (i.e. the existence of TBCs \cite{davydov2011wittgroupnondegeneratebraided}).
Now we have the following generalization of (\ref{vectorization})
\be \operatorname{Hom}(\mathcal H_{\mathcal T'},\mathcal H_{\mathcal T})\to \mathcal H_{\mathcal T\times\bar{\mathcal T}'}:~~|a\rangle\langle b|\mapsto |a,b\rangle,\quad a\in\sG^g,b\in{\sG'}^g.\label{vectorization2}\ee The topological boundary conditions of $\mathcal T\times\bar{\mathcal T'}$ are mapped into interfaces between the two theories.

\subsection{Construction and classification of surface operators}\label{class1}
Surface operators of $\mathcal T=(\sG,q)$ can be constructed by higher-gauging subgroups of $\sG$ \cite{Roumpedakis:2022aik}.
The $F$-symbols of the anyons in $\sG$ can be chosen to be trivial so that any subgroup $\H\subseteq\sG$ can be gauged on a codimension-1 surface. The possible gaugings are classified by pairs $(\H,\nu)$, where $\H\subseteq \sG$ and $\nu\in H^2(\H,\mathbb C^\times)$ is a choice of discrete torsion.  To gauge $(\H,\nu)$, we insert a network of anyon lines dual to a triangulation of the surface and sum over all charges in $\H$. 
A choice of discrete torsion is equivalent to the choice of an antisymmetric form $B:\H\times\what\H\to\mathbb Z_N$ and we can write $\nu(a,b)=\omega^{2^{-1}B(a,b)}$. Since a genus-$g$ surface operator is the $g$-fold tensor product of a genus-$1$ surface operator, we only need to consider the sum on a torus
\be \begin{split}\rho_{{\H,\nu}}&={1\over |\H|}\sum_{\alpha\in  H_1(\Sigma_1,\H)}\nu(\alpha)W(\alpha)\\
	&={1\over |\H|}\sum_{a,b\in\H}\sum_{x\in\sG}\omega^{2^{-1}B(a,b)}\omega^{2^{-1}g(a,b)}|x+a\rangle\langle x|\omega^{g(b,x)}\\
	&={1\over |\H|}\sum_{a,b\in\H}\sum_{x\in\sG}\omega^{(2^{-1}aB+2^{-1}ag+xg)b^T}|x+a\rangle\langle x|\\
	&=\sum_{a\in\H}\sum_{x\in\sG}\delta_{2^{-1}a(B+g)+xg\in\H^\perp}|x+a\rangle\langle x|\\
	&=\sum_{a\in\H,h\in \H^\perp}|a(-\id+Bg^{-1})+h\rangle\langle a(\id+Bg^{-1})+h|,\end{split}\label{rho}\ee
where in the third line and afterwards we use the notation $g(a,b)=ag b^T$ for a bilinear form. In the fourth line, we denote by $\H^\perp$ the dual module to $\H$ with respect to $g$. This expression immediately tells us that the surface operator can be obtained by summing over the submodule $\mathcal L\subset\G$ 
\be \rho_{{\H,\nu}}=\sum_{(a,b)\in \mathcal L}|a\rangle\langle b|,~~\mathcal L=\begin{pmatrix}
	\H^\perp & \H^\perp\\
	\H(Bg^{-1}-\id)& \H(Bg^{-1}+\id)
\end{pmatrix},\label{Cgen}\ee
where we wrote the generator of $\mathcal L$ in a block form and by abuse of notation we use the symbols $\mathcal L$, $\H$ and $\H^\perp$ to denote a generating matrix for these modules. The module $\mathcal L$ is automatically Lagrangian with respect to the quadratic form $\q=q\oplus(-q)$ \cite{KS2}.

Gauging $\H$ on a surface results in a ``quantum symmetry" $\what\H\times\Q$, where $\Q=\sG/\H$, of lines confined to the surface (see appendix \ref{highergauging}). A subtle point is that depending on $\H$, the following exact sequence may not split
\be  \begin{tikzcd}
	0 \arrow[r] & \H \arrow[r] & \sG  \arrow[r] & \Q \arrow[r] & 0 
\end{tikzcd}.\ee
Such extensions are classified by 2-cocycles in the group cohomology $\mu\in H^2(\Q,\H)$. The trivial $\mu$ means that $\sG\cong \H\times\Q$. A non-trivial $\mu$ results in a higher symmetry with a mixed anomaly \cite{Tachikawa_2020}. This anomaly prevents gauging of arbitrary subgroups of $\what H\times\Q$, but subgroups of $\what H$ or $\Q$ can still be gauged, with some choice of discrete torsion, to obtain another surface operator. This motivates the definition of a groupoid \cite{Gaiotto_2021}, whose objects are surface operators and morphisms are gauging operations. In particular, the groupoid structure implies that all surface operators can be obtained by gauging the trivial surface, in other words (\ref{Cgen}) gives the \textit{full} set of $\q$-Lagrangian submodules. 

When the exponent of $\sG$ is squarefree, the groupoid structure simplifies. For any subgroup $\H$ we necessarily have $\H\times\Q\cong \sG$ and any subgroup can be again gauged. This is reflected in the fact that $\O(\G)$ acts transitively on the Lagrangian subgroups, hence all gauging operations can be realized by the global action of $\O(\G)$. In particular, all surface operators can be obtained by acting on the trivial surface by $\O(\G)$, hence the objects of the groupoid are in bijection with $\O(\G)/\Gamma_{\vec 1}$, where $\Gamma_{\vec 1}$ is the subgroup that stabilizes the trivial surface. 

When the exponent of $\sG$ is not squarefree, not all Lagrangians of the form (\ref{Cgen}) belong to the same $\O(\G)$-orbit, since they are non-isomorphic as groups. Instead, surface operators are split into multiple $\O(\G)$-orbits.

To describe this classification, we focus on $\G=\mathbb Z_N^{2n}$ and specialize the discussion from \ref{sec:classification} to surface operators by bringing the bilinear form $\g$ from (\ref{antidiag}) into the diagonal form
\be \g=\begin{pmatrix}
	\id_{n\times n}& 0\\
	0 & -\id_{n\times n}
\end{pmatrix}.\ee
Theorem \ref{theorem:Orthogonal_codes} now gives a classification of surface operators. Define the tuples $\vec q$ as in theorem \ref{theorem:Orthogonal_codes}.
Then, the generator of any Lagrangian submodule $\mathcal L\subset\sG\times\bar\sG=\G$ with respect to $\q=q\oplus(-q)$ can be written as
\be \begin{pmatrix}
	\operatorname{diag}(\vec q)& \operatorname{diag}(\vec q)\\
	\operatorname{diag}(N/\vec q) & \operatorname{diag}(\vec 0)
\end{pmatrix}\times Q,\label{gen2}\ee
for some  $Q\in \O(\G)$. Above we defined $(N/\vec q)=(N/q_1,\dots,N/q_n)$.

The generator matrix (\ref{gen2}) for $Q=\id$ is in block form and we can write the surface operator explicitly in the computational basis
\be \rho_{\vec q}= \sum_{(a,b)\in \mathbb Z_{N/\vec q}\times\mathbb Z_{\vec q}}\bigotimes_{i=1}^{n}|a_iq_i+b_i {N/ q_i}\rangle\langle a_iq_i|.\label{rhoi}\ee
For each $\vec q\in\mathbb S$ we define the subgroup $\Gamma_{\vec q}\subseteq \O(\G)$ that stabilizes the submodule (\ref{gen2}) with $Q=\id$. Then, the cosets $\O_{\vec q}=\O(\G)/\Gamma_{\vec q}$ are in bijection with the surface operators in the orbit labeled by $\vec q$.  An arbitrary surface operator is uniquely parametrized by $\vec q\in\mathbb S$ and $Q\in \O_{\vec q}$ as $\rho_{\vec d,Q}$.
Note that the surface operators with $\vec d\neq\vec 1$ are always non-invertible. Operators $\rho_{\vec 1,Q}$ may be invertible, depending on $Q$. In particular $\rho_{\vec 1,\id}$ is the trivial (identity) operator.


Surface operators form a non-commutative monoid graded by $\mathbb S$. Fusing two operators $\rho_{\vec q,Q}$, $\rho_{\vec q',Q'}$ results into an operator with grade $\vec q''=\vec q\lor\vec q'$. The set of invertible operators form a group $\mathbb G$, which is a faithful unitary representation of $\O(\sG)$. Such operators are obtained from Lagrangian submodules with generators of the form $(\id|\mathsf B)$, where $\mathsf B \in\O(\sG)$ and they all belong to the orbit $\O_{\vec 1}$. The identity element of $\mathbb G$ is the trivial surface operator $\rho_{\vec 1,\id}$, described by the module $(\id|\id)$, while the surface $\rho_{\vec 1,-\id}$ is the charge-conjugation operation, described by $(\id|-\id)$. From the CFT point of view, the invertible operators describe $T$-dualities, preserving the (unflavored) Siegel-Narain partition function.

\subsection{Construction of interfaces}\label{qiso}

In this subsection, we describe the construction of interfaces between $\mathcal T$ and $\mathcal T'=\mathcal T/\H$, for some isotropic $\H\subseteq\sG$. The diagonal interface $\mathcal I_\H$ can be constructed by starting with $\mathcal T$ on $\Sigma_g\times I$ and gauging the non-anomalous subgroup $\H$ in half of the space \cite{KNBalasubramanian:2025vum}. Then, the rest of the interfaces can be obtained by fusing $\mathcal I_\H$ with surface operators of $\mathcal T$.

According to the discussion in \ref{gaug}, a $q$-isotropic submodule $\H\subset\sG$ is non-anomalous and can be gauged in a handlebody. The resulting theory $\mathcal T/\H$ has spectrum of lines in $\K=\H^\perp/\H$ and torus Hilbert space with basis $\{|a\rangle_{\mathcal T/\H},a\in\K\}$. 
The interface $\mathcal I_{\H}$ between the original and the gauged theory can be constructed by placing $\mathcal T$ on $\Sigma_1\times [0,1]$ and gauging $\H$ in half of the space $\Sigma_1\times[0,1/2]$. This operator is an isometric embedding $\mathcal I_{\H}:\mathcal H_{\mathcal T/\H}\to\mathcal H_{\mathcal T}$, which we can write explicitly in the anyon bases of the two theories
\be \mathcal I_{\H}={1\over \sqrt{|\H|}}\sum_{a\in\K}\sum_{h\in\H}|h+\pi^{-1}(a)\rangle\langle a|_{\mathcal T/\H} \label{enc},\ee
where $\pi:\H^\perp\to\K$ is the canonical projection. The conjugate map $\mathcal I_{\H}^\dagger:\mathcal H_{\mathcal T}\to \mathcal H_{\mathcal T/\H}$ is a projection onto the Hilbert space of $\mathcal T/\H$. The genus-$g$ version of this operator is simply the $g$-fold tensor product of the torus operator.

The image of $\mathcal I_{\H}$ at genus-$g$ is the code subspace of the stabilizer corresponding to the module $\H\otimes\mathbb Z_N^{2g}$ and the Hilbert space of $\mathcal T/\H$ can be viewed as the space of logical states. Therefore $\mathcal I_\H$ and $\mathcal I_{\H}^\dagger$ are encoding/decoding maps. We emphasize here that the encoding (\ref{enc}) embeds a system of $g$ ``logical" qudits into another system of $g$ ``physical" qudits with larger local dimensions of their Hilbert space, rather than a system of more qudits as is typical in quantum error correction. 


\section{Examples}\label{sec:examples}

In this section we consider various examples in which we explicitly construct all TBC states by applying the method of section \ref{sec3} and then calculate the $\lambda$-matrix and TQFT gravity state (\ref{gravityZdef}). We emphasize that even though we choose a specific quadratic form to calculate the $\q$-Lagrangian submodules, the classification of TBC states and the $\t$-matrix are the same for all equivalent TQFTs.

\subsection{Example I: $\sG=\mathbb Z_{p^m}$}\label{sec:ex1}

We consider theories of the form $\mathcal T \times \bar{\mathcal T}'$, where $\mathcal T=(\sG,q)$ for a cyclic group $\sG$ and $\mathcal T'=\mathcal T/\H$ for some isotropic $\H\subseteq\sG$. We set $\sG=\mathbb Z_{p^m}$ for some odd prime $p$.

A concrete realization of the $\mathcal T\times \bar{\mathcal T}$ theory is given by the $U(1)_{p^m}\times U(1)_{-p^m}$ CS theory studied in \cite{Raeymaekers:2021ypf,Raeymaekers_2024} or its level-rank dual $SU(p^m)_1\times SU(p^m)_{-1}$ \cite{AngelinosWZW}. Moreover, by an automorphism of $\G$, this theory becomes the untwisted Dijkgraaf-Witten theory studied recently in \cite{Nicosanti:2025xwu}, which is the same as the level-$p^m$ toric code.

\subsubsection{Surface operators}

The group $\sG=\mathbb Z_{p^m}$ has $m+1$ subgroups $\H_k=p^{m-k}\mathbb Z_{p^m}\cong \mathbb Z_{p^k}$ that can be gauged on a surface. Since $H^2(\mathbb Z_N,\mathbb C^\times)$ is trivial, there is a unique surface operator for each subgroup. Equivalently, each surface  operator corresponds to a Lagrangian  $\mathcal L_k\subseteq\sG\times \bar\sG=\G$ with respect to the quadratic form $q\oplus (-q)=\q$. By applying (\ref{Cgen}) we obtain
\be \mathcal L_k=\begin{cases}
	\begin{pmatrix}
		p^{k} & p^k\\
		0 & p^{m-k}
	\end{pmatrix} & 0\leq k\leq \lfloor {m\over 2}\rfloor\\
	\begin{pmatrix}
		p^{m-k} & -p^{m-k}\\
		0 & p^{k}
	\end{pmatrix} & \lfloor {m\over 2}\rfloor<k\leq m.
\end{cases}\label{lagr}\ee

The genus-$g$ surface operators of $\mathcal T$ map to genus-$g$ TBC states $|\mathcal L_i\rangle$ of $\mathcal T\times\bar{\mathcal T}$.

\subsubsection{Gravity State of $\mathcal T\times\bar{\mathcal T}$}

We now calculate the $\t$-matrix and the gravity state for $\mathcal T\times\bar{\mathcal T}$ using the formula (\ref{gravityZ}).
The number of symplectic orbits $\mathbb T$ at genus $g$ is given by proposition \ref{orbitscount}
\be |\mathbb T|={(\lfloor {m\over 2}\rfloor +g)!\over (\lfloor {m\over 2}\rfloor )!g!},\ee
while for the orbits of TBC
\be |\mathbb S|=\left\lfloor{m\over 2}\right\rfloor+1 .\ee
These two are equal only at genus $1$, hence the $\t$-matrix is square only in that case.

The intersection matrix $D$ is Toeplitz
\be D_{ij}={1\over p^{mg}}|\mathcal L_i\cap \mathcal L_j|^g={1\over p^{g|i-j|}}\ee
and its inverse is tridiagonal
\be (D^{-1})_{ij}= {p^g\over p^{2g}-1}\left((p^g+p^{-g}) I-\tilde I p^{-g}-\epsilon\right),\ee
where $I$ is the identity matrix, $\tilde I_{ij}=\delta_{ij}(\delta_{i0}+\delta_{im})$ and $\epsilon_{ij}=\delta_{|i-j|,1}$. Note that $\mathcal L_i$ and $\mathcal L_{m-i}$ belong to the same $\O(\G)$-orbit. The matrix $Y$ is therefore
\be Y={1\over p^g-p^{-g}}(2(p^g+p^{-g})I-2\epsilon-\hat I) ,\label{Tt}\ee
where $\hat I_{ij}=-2p^{-g}\delta_{i0}\delta_{j0}$ if $m$ is even and $\hat I_{ij}=-2p^{-g}\delta_{i0}\delta_{j0}-(p^g+p^{-g})\delta_{ij}\delta_{i,{m-1\over 2}}$ if $m$ is odd.

The next step is to calculate the ``seed" matrix (\ref{vmatr})
\be v_{p^i}^{\vec d}=\prod_{j=1}^g \gcd(p^i,d_j) \ee
and finally the $\t$-matrix
\be \t_{\vec d,p^i}=\sum_{j=0}^{\lfloor m/2\rfloor+1}v_{p^j}^{\vec d}Y_{j,i}.\label{tT}\ee

\paragraph{Genus 1:} At genus $g=1$, the $\t$-matrix is square. Using that $\vec d\in\{1,p,\dots,p^{\lfloor {m\over 2}\rfloor}\}$ we calculate the seed vector
\be v_{p^i}^{p^j}=p^{\min\{i,j\}} .\ee
Using (\ref{Tt}) with (\ref{tT}), we find the matrix $\t_{p^i,p^j}$, which is upper triangular. For odd $m$
\be \t_{p^i,p^j}={2p^{i+1}\over p+1}\begin{cases}
	1 & i=j\\
	{p-1\over p}& i<j\\
	0 & i>j,
\end{cases}\label{ttodd}\ee
while for even $m$
\be \t_{p^i,p^j}={2p^{i+1}\over p+1}\begin{cases}
	1 & i=j<{m\over 2}\\
	{p-1\over p}& i<j<{m\over 2}\\
	{p-1\over 2p} & i<{m\over 2},j={m\over 2}\\
	{p+1\over 2p} & i=j={m\over 2}\\
	0 & i>j.
\end{cases}\label{tteven}\ee
Explicitly, the gravity state, for arbitrary bulk weights $w_{p^i}$, is given by
\be \mathcal Z[\mathcal T\times\bar{\mathcal T},\Sigma_1;w]=\begin{cases}\sum_{i=0}^{\lfloor m/2\rfloor} w_{p^i}{p^{i+1}\over p+1}\left(|\mathcal L_i\rangle+|\mathcal L_{m-i}\rangle+{p-1\over p}\sum_{j=i+1}^{m-i-1}|\mathcal L_j\rangle\right) & m\text{ odd}\\
\sum_{i=0}^{m/2-1} w_{p^i}{p^{i+1}\over p+1}\left(|\mathcal L_i\rangle+|\mathcal L_{m-i}\rangle+{p-1\over p}\sum_{j=i+1}^{m-i-1}|\mathcal L_j\rangle\right)+w_{p^{m/2}}p^{m/2}|\mathcal L_{m/2}\rangle & m\text{ even}	.\end{cases}\label{ZTQFT1}\ee

The genus-reduction prescription \ref{DSpresc} amounts to fixing the bulk weights such that the coefficient of each TBC state that appears in (\ref{ZTQFT1}) is equal to $1$. This results in the following bulk weights
\be w_{p^i}=\begin{cases}{p+1\over p^{2i+1}} & 0\leq i<m/2\\ p^{-m} & i=m/2.\end{cases}\ee

\paragraph{Genus $>1$:}
For $g>1$ the $\t$-matrix is no longer square and different choices of bulk weights can lead to the same result. This happens because the $\eta$-Lagrangian stabilizer states become linearly dependent. We will calculate the gravity state by setting some bulk weights to zero. Specifically, define the ``diagonal" weights $\tilde w_{p^i}\equiv w_{(p^i,\dots,p^i)}$ and set all other weights to zero. With this choice, the calculation is identical to the genus-$1$ case after the substitution $p\mapsto p^g$. We proceed by restricting the $\t$-matrix and seed vector $v$ to the indices on which the weights are now supported. The restricted seed vector is
\be \bar v_{p^i}^{p^j}=\prod_{l=1}^g p^{\min\{i,j\}}=p^{g\min\{i,j\}}.\ee
The restricted $\bar \t$-matrix then reads (for simplicity we present it only for odd $m$)
\be \bar \t_{p^i,p^j}={p^{g(i+1)}\over p^g+1}\begin{cases}
	1 & i=j\\
	{p^g-1\over p^g}& i<j\\
	0 & i>j.
\end{cases}\ee
The gravity state (for odd $m$) is given by
\be \mathcal Z[\mathcal T\times\bar{\mathcal T},\Sigma_g;\tilde w]=\sum_{i=0}^{\lfloor m/2\rfloor} \tilde w_{p^i}{p^{gi+g}\over p^g+1}\left(|\mathcal L_i\rangle+|\mathcal L_{m-i}\rangle+{p^g-1\over p^g}\sum_{j=i+1}^{m-i-1}|\mathcal L_j\rangle\right).\label{TQFTZ3}\ee
The even $m$ expression can be obtained by substituting $p\mapsto p^g$ in (\ref{ZTQFT1}).

The prescription \ref{DSpresc} fixes the remaining bulk weights as follows
\be \tilde w_{p^i}=\begin{cases}{p^g+1\over p^{g(2i+1)}} & 0\leq i<m/2\\ p^{-gm} & i=m/2.\end{cases}\ee

\subsubsection{Interfaces}

The quadratic form $q$ of $\sG$ vanishes on the subgroups $\H_k=p^{m-k}\mathbb Z_{p^m}\cong \mathbb Z_{p^{k}}$ for $k\leq m/2$, hence these are the only subgroups that can be gauged. They are $\lfloor {m\over 2}\rfloor+1$ in total, labeled by $k=0,1,\dots,\lfloor {m\over 2}\rfloor$.

To gauge $\H_k$ \cite{Kaidi:2021gbs}, we remove the lines that do not braid trivially with all lines in $\H_k$ and we identify lines in $\H_k$. The resulting chiral theory has residual $1$-form symmetry given by $\K_k=\H_k^\perp/\H_k\cong \mathbb Z_{p^{m-2k}}$.
Summing over all insertions of the lines in $H_k$ on the handlebody $\mathcal V_g$ we obtain a state that is a $g$-fold product of a torus state
\be |\Lambda_k\rangle^{\otimes g}={1\over \sqrt{H_1(\mathcal V_g,\H_k)}}\sum_{a\in H_1(\mathcal V_g,\H_k)}W(a)|0\rangle^{\otimes g},\ee
where after a short calculation
\be |\Lambda_k\rangle={1\over p^{k/2}}\sum_{a\in\mathbb Z_{p^k}}|ap^{m-k}\rangle.\label{omegak}\ee
Inserting lines with charges in $\H_k^\perp/\H_k\cong\mathbb Z_{p^{m-2k}}$ we can build a subspace isomorphic to the Hilbert space of $\mathcal T/\H_k$. This is the code subspace associated with $\H_k\otimes \mathbb Z_N^2$, with basis
\be X_{ap^k}|\Lambda_k\rangle={1\over p^{k/2}}\sum_{b\in\mathbb Z_{p^k}} |ap^k+bp^{m-k}\rangle \label{bulkHS}.\ee
The interface $\mathcal I_k:\mathcal H_{\mathcal T/\H_k}\to\mathcal H_{\mathcal T}$ is therefore given by
\be \mathcal I_k={1\over p^{k/2}}\sum_{b\in\mathbb Z_{p^k}} |ap^k+bp^{m-k}\rangle\langle a|_{\mathcal T/\H_k} .\label{interf}\ee
With this process we built the diagonal interface between $\mathcal T=(\mathbb Z_{p^{m}},q)$ and $\mathcal T'=(\mathbb Z_{p^{m-2k}},q)$. The full theory $\mathcal T\times\bar{\mathcal T}'$ has $1$-form symmetry group $\G=\mathbb Z_{p^m}\times\bar{\mathbb Z}_{p^{m-2k}}$ and quadratic form $\q(a,b)=a^2-p^{2k}b^2$.

The rest of the interfaces can be obtained by fusing (\ref{interf}) with surface operators of $\mathcal T$ (\ref{lagr}). 
For fixed $k$, the distinct interfaces at genus-$1$ are as follows
\be \tilde\rho_i\equiv p^{-k/2}\rho_i\mathcal I_k=\begin{cases}
	\sum |p^{m-i}a+p^ib\rangle\langle p^{i-k}b| & m/2\geq i\geq k\\
	\sum |p^{i}a-p^{m-i}b\rangle\langle p^{m-i-k}b| & m-k\geq i> m/2 \\
	0 & \text{otherwise}
\end{cases}\ee
and we denote the TBC state corresponding to $\tilde \rho_i^{\otimes g}$ by $|\tilde{\mathcal L}_i\rangle$. Explicitly, the Lagrangian submodules $\tilde{\mathcal L}_i\subset\mathbb Z_{p^m}\times\bar{\mathbb Z}_{p^{m-2k}}$ are generated by
\be \tilde{\mathcal L}_i=\begin{cases}
	\begin{pmatrix}
		p^{i} & p^{i-k}\\
		p^{m-i} & 0
	\end{pmatrix} & k\leq i\leq \lfloor {m\over 2}\rfloor\\
	\begin{pmatrix}
		p^{m-i} & -p^{m-i-k}\\
		p^{i} & 0
	\end{pmatrix} & \lfloor {m\over 2}\rfloor<i\leq m-k.
\end{cases}\label{lagr2}\ee


\subsubsection{Gravity states of $\mathcal T\times{\bar {\mathcal T}}'$}
From (\ref{lagr2}), the intersection matrix is again Toeplitz
\be D_{i,j}={|\tilde{\mathcal L}_i\cap \tilde{\mathcal L}_j|^g\over p^{g(m-k)}}=p^{-g|i-j|},~~k\leq i,j\leq m-k.\ee
The Lagrangian submodules are partitioned into $\lfloor{m-2k\over 2}\rfloor+1$ orthogonal orbits. The elements of the matrix $Y$ are again given by (\ref{Tt}), but resized into an $(\lfloor {m-2k\over 2}\rfloor+1)\times (\lfloor {m-2k\over 2}\rfloor+1)$ matrix. 

\paragraph{Genus 1:}
At genus $1$, using (\ref{stabstate}), the seed matrix is given by
\be v_{p^j}^{p^i}\equiv \langle \Omega_{p^i}|\mathcal L_j\rangle=p^{\min\{i,j\}}p^{-\min\{i/2,k\}},~~ 0\leq i\leq \lfloor m/2\rfloor,~k\leq j\leq \lfloor m/2\rfloor.\ee
The extra branched factor occurs because the right sector $\mathcal T/\H_k$ of the theory is unable to distinguish between topologies with $\vec d\in\{1,p,\dots,p^{\min\{2k,\lfloor m/2\rfloor-k\}}\}$, seeing them all as handlebodies, while the left sector $\mathcal T$, having a larger group exponent, is able to detect their structure. This factor can be absorbed by changing the normalization of the states (\ref{stabstate}).

The $\eta$-Lagrangian orbit averages $\avg{|\Omega_{\vec d}\rangle}$ are now linearly dependent and thus the $\t$-matrix is not uniquely defined. These linear relations can be obtained from the left nullspace of $v$ and using these relations, $\t$ can be brought to a form where its first $k$ rows are zero. Rescaling the states $|\Omega_{p^i}\rangle$, by choosing a more convenient normalization, the remaining square submatrix of $\t$ has entries given by (\ref{ttodd}-\ref{tteven}), but with both of its indices running from $k$ to $\lfloor m/2\rfloor$.

The gravity state at genus $1$ now reads (we only write the expression for odd $m$)
\be \mathcal Z[\mathcal T\times\bar{\mathcal {T}}',\Sigma_1;w]=\sum_{i=k}^{\lfloor m/2\rfloor} w_{p^{i}}{p^{i+1}\over p+1}\left(|\tilde{\mathcal L}_i\rangle+|\tilde{\mathcal L}_{m-i}\rangle+{p-1\over p}\sum_{j=i+1}^{m-i-1}|\tilde{\mathcal L}_j\rangle\right).\label{tqintr}\ee
In fact, there is a much easier way to obtain the gravity state for this theory. Starting from (\ref{ZTQFT1}) we can apply the inverse ``vectorization" map (\ref{vectorization2}) to turn it into an operator in $\operatorname{Hom}(\mathcal H_{{\mathcal T}'},\mathcal H_{\mathcal T})$. Fusing this operator with the interface $\mathcal I_{k}$ followed by (\ref{vectorization2}), we immediately obtain (\ref{tqintr}), up to an overall constant factor.

\paragraph{Genus $>$ 1:}
More generally, we can apply the same process as in (\ref{TQFTZ3}) to obtain a genus-$g$ state restricting to diagonal weights $\tilde w_{p^i}=w_{(p^i,\dots,p^i)}$. This amounts to replacing $p\mapsto p^g$ in (\ref{tqintr}) and we obtain
\be \mathcal Z[\mathcal T\times\bar{\mathcal T}',\Sigma_g;\tilde w]=\sum_{i=k}^{\lfloor m/2\rfloor} \tilde w_{p^{i-k}}{p^{g(i+1)}\over p^g+1}\left(|\tilde{\mathcal L}_i\rangle+|\tilde{\mathcal L}_{m-i}\rangle+{p^g-1\over p^g}\sum_{j=i+1}^{m-i-1}|\tilde{\mathcal L}_j\rangle\right).\label{TQFTZ4}\ee

\subsection{Example II: $\sG=\mathbb Z_{p^2}\times\mathbb Z_{p^2}$.}\label{sec:ex2}

In this example, we consider the TQFT $\mathcal T\times\bar{\mathcal T}$, where $\mathcal T=(\sG,q)$ with $\sG=\mathbb Z_{p^2}^2$ and $q(a,b)=a^2+b^2$. We also consider all theories of the form $\mathcal T\times\bar{\mathcal T}'$, where $\mathcal T'=\mathcal T/\H$ for $\H$ an isotropic subgroup of $\sG$.

\subsubsection{Surface operators}
There are $3p^2+4p+2$ Lagrangian submodules of $\G=\sG\times\bar\sG$, where $\sG=\mathbb Z_{p^2}^2$. They can be obtained (using (\ref{Cgen})) by gauging $(\H,\nu)$ on a surface, where $\H\subseteq\sG$ and $\nu\in H^2(\H,\mathbb C^\times)$ is a choice of discrete torsion. The complete classification is as follows:
\begin{enumerate}
	\item $\H=\{0\}$: This leads to the Lagrangian submodule describing the trivial surface 
	\be \begin{pmatrix}
		1 & 0 & 1 & 0\\
		0 & 1 & 0 & 1
	\end{pmatrix}.\ee
	\item $\H\cong\mathbb Z_p$: There are $p+1$ subgroups in this isomorphism class, with generators $(p,0)$, $(pa,p)$ for $a\in\{0,1\dots,p-1\}$. The Lagrangian submodules are
	\be \begin{pmatrix}
		p & 0 & p & 0\\
		0 & 1 & 0 & 1\\
		p& 0 &-p & 0
	\end{pmatrix},~~\begin{pmatrix}
		1 & -a & 1 & -a\\
		0 & p & 0 & p\\
		pa & p & -pa & -p
	\end{pmatrix},\quad a\in\{0,1\dots,p-1\}.\ee
	\item $\H\cong\mathbb Z_{p^2}$: There are $p^2+p$ subgroups with generators $(a,1)$ for $a\in\mathbb Z_{p^2}$ and $(1,bp)$ with $b\in\mathbb Z_{p}$, with Lagrangians
	\be \begin{pmatrix}
		1 & -a & 1 & -a\\
		a & 1 & -a & -1
	\end{pmatrix}~~\begin{pmatrix}
		bp & -1 & bp & -1\\
		1 & bp & -1 & -bp
	\end{pmatrix},~~a\in\mathbb Z_{p^2},~b\in\mathbb Z_p.\ee
	\item $\H\cong\mathbb Z_{p}\times \mathbb Z_p$: There is $1$ subgroup, generated by $\langle (p,0),(0,p)\rangle$ and $|H^2(\mathbb Z_{p}\times \mathbb Z_p,\mathbb C^\times)|=p$, resulting in $p$ Lagrangians
	\be \begin{pmatrix}
		p & 0 & p & 0\\
		0 & p & 0 & p\\
		-p & -k& p&-k\\
		k&-p & k & p
	\end{pmatrix},~~k\in\{0,1\dots,p-1\}.\ee
	\item $\H\cong\mathbb Z_{p^2}\times \mathbb Z_p$: There are $p+1$ subgroups generated by $\langle (0,1),(p,0)\rangle$ and $\langle (1,a),(0,p)\rangle$ for $a\in\{0,\dots,p-1\}$. Moreover, $|H^2(\mathbb Z_{p^2}\times \mathbb Z_p,\mathbb C^\times)|={p}$ for a total of $p^2+p$ Lagrangians
	\be \begin{pmatrix}
		p & 0 & p & 0\\
		k & 1 & k & -1\\
		p & -pk &-p &-pk
	\end{pmatrix},~~\begin{pmatrix}
		ap & -p & ap & -p\\
		ka-1&-k-a&ka+1&-k+a\\
		pk &-p & pk & p
	\end{pmatrix},~~a,k\in\{0,1\dots,p-1\} .\ee
	\item $\H\cong\mathbb Z_{p^2}\times \mathbb Z_{p^2}$: There is a unique subgroup and $|H^2(\mathbb Z_{p^2}\times \mathbb Z_{p^2},\mathbb C^\times)|=p^2$ 
	\be \begin{pmatrix}
		-1 & -k & 1 & -k\\
		k &-1 & k & 1
	\end{pmatrix},~~k\in\mathbb Z_{p^2}.\ee
\end{enumerate}

The Lagrangian submodules are partitioned into the $\O(\G)$-orbits $\O_{\vec q}$, with $\vec q\in \{(1,1),(1,p),(p,p)\}$. By direct counting, the orbit $\O_{(1,1)}$ contains $2p(p+1)$ submodules. Since $\O_{(1,1)}$ contains all the invertible surface operators, it is useful to compare it with $\O(\sG)$.
To count the elements of $\O(\sG)$, consider the modulo $p$ projection $\pi:\O(\mathbb Z_{p^2}^2)\to \O(\mathbb F_p^2)$. The kernel of $\pi$ consists of $2\times 2$ matrices of the form $\id+pX$, where $X^T=-X$, therefore $|\ker\pi|=p$. By the first isomorphism theorem, 
\be |\O(\mathbb Z_{p^2}^2)|=|\ker\pi||\O(\mathbb F_p^2)|=\begin{cases}
	2p(p+1) & p=3\mod 4\\
	2p(p-1) & p=1\mod 4.
\end{cases}\ee 
Therefore, if $p=3\mod 4$ the orbit $\O_{(1,1)}$ consists precisely of the invertible surface operators implementing $\O(\sG)$. Otherwise, $\O_{(1,1)}$ contains $4p$ additional surface operators that are non-invertible.

The rest of the orbits contain $|\O_{(1,p)}|=(p+1)^2$ and $|\O_{(p,p)}|=1$ non-invertible operators.

\subsubsection{Gravity state of $\mathcal T\times\bar{\mathcal T}$}

This time we have $|\mathbb S|=3$ orbits with $\mathbb S=\{(1,1),(1,p),(p,p)\}$ and $|\mathbb T|=g+1$ at genus-$g$.
The first task is to obtain matrix $Y_{\vec q,\vec q'}$. To achieve that, we calculate the following matrix
\be \mathcal D_{\vec q,\vec q'}\equiv\sum_{Q\in \O_{\vec q},Q'\in\O_{\vec q'}} D_{\vec q,Q;\vec q',Q'}.\ee
Using the ordering $(1,1)$, $(1,p)$, $(p,p)$ for the indices we can write it as
\be \mathcal D=\begin{pmatrix}
	2p(p+1)(1+p^{2(1-2g)}+p^{-2g}(p^2+2p-1)) & 2(p+1)^2 p^{1-3g}(p+p^{2g})& 2(p+1)p^{1-2g}\\
	2(p+1)^2 p^{1-3g}(p+p^{2g})& (p+1)^2(1+(p+2)p^{1-2g})& (p+1)^2p^{-g}\\
	2(p+1)p^{1-2g} &(p+1)^2p^{-g}& 1
\end{pmatrix},\ee
from which we obtain
\be Y= \left(
\begin{array}{ccc}
	\frac{2 (p+1) p^{4 g+1}}{\left(p^{2 g}-1\right) \left(p^{2 g}-p^2\right)} & -\frac{2 (p+1)^2 p^{3 g+1}}{\left(p^{2 g}-1\right)
		\left(p^{2 g}-p^2\right)} & \frac{2 (p+1) p^{2 g+2}}{\left(p^{2 g}-1\right) \left(p^{2 g}-p^2\right)} \\
	-\frac{2 (p+1)^2 p^{3 g+1}}{\left(p^{2 g}-1\right) \left(p^{2 g}-p^2\right)} & \frac{(p+1)^2 p^{2 g} \left(p^{2 g}+(p+2)
		p\right)}{\left(p^{2 g}-1\right) \left(p^{2 g}-p^2\right)} & -\frac{(p+1)^2 p^g \left(p^{2 g}+p^2\right)}{\left(p^{2 g}-1\right)
		\left(p^{2 g}-p^2\right)} \\
	\frac{2 (p+1) p^{2 g+2}}{\left(p^{2 g}-1\right) \left(p^{2 g}-p^2\right)} & -\frac{(p+1)^2 p^g \left(p^{2
			g}+p^2\right)}{\left(p^{2 g}-1\right) \left(p^{2 g}-p^2\right)} & \frac{p^{4 g}+2 p^{2 g+1}+p^4}{\left(p^{2 g}-1\right)
		\left(p^{2 g}-p^2\right)} \\
\end{array}
\right).\ee

\paragraph{Genus 1:}
At genus $1$, the TBC states are linearly dependent and $Y$ becomes singular. However, we can calculate the $\t$-matrix with $g$ as a parameter and then send $g\to1$ to obtain
\be \t=\left(
\begin{array}{ccc}
	p^2 & 0 & 0 \\
	-p^3 & p^2 (p+1) & 0 \\
\end{array}
\right).\ee
Note that $\t$ contains negative entries. This happens because the TBC states are linearly dependent at genus $1$, which results in the linear relation between the three $\O(\G)$-averaged states
\be p|\overline{\mathcal L_{(1,1)}}\rangle-(p+1)|\overline{\mathcal L_{(1,p)}}\rangle+|\overline{\mathcal L_{(p,p)}}\rangle=0,\quad g=1,\ee
 and thus the $\t$-matrix is not uniquely fixed. Using this linear relation, $\t$ can be brought to an upper-triangular form with nonnegative entries
\be \t=\left(
\begin{array}{ccc}
	p^2 & 0 & 0 \\
	0 & 0 & p^{2} \\
\end{array}
\right).\ee
The gravity state can finally be expressed as
\be \mathcal Z[\mathcal T\times\bar{\mathcal T},\Sigma_1;w]=p^2\left(w_{1}|\overline{\mathcal L_{(1,1)}}\rangle+w_{p}|\overline{\mathcal L_{(p,p)}}\rangle\right).\label{ZTQFT3}\ee

\paragraph{Genus 2:}
At $g=2$, the $\t$-matrix is square. Using the ordering $(1,1)$, $(1,p)$, $(p,p)$ for its indices we write it as an upper triangular matrix
\be t={p^2\over p^2+1}\left(
\begin{array}{ccc}
	2 p^2 & p^2-1 & 0 \\
	0 & p^2 (p+1) & p^2-p \\
	0 & 0 & p^2 \left(p^2+1\right) \\
\end{array}
\right).\ee
The genus-$2$ gravity state can now be expressed as
\be \begin{split}\mathcal Z[\mathcal T\times\bar{\mathcal T},\Sigma_2;w]=&{p^2\over 1+p^2}\bigg(2p^2 w_{(1,1)}|\overline{\mathcal L_{(1,1)}}\rangle+p^2(p^2+1)w_{(p,p)}|\overline{\mathcal L_{(p,p)}}\rangle\\& +(w_{(1,1)}(p^2-1)+w_{(1,p)}p(p-1))|\overline{\mathcal L_{(1,p)}}\rangle\bigg).\end{split}\label{ZTQFT4}\ee

\subsubsection{Interfaces and Gravity States of $\mathcal T\times\bar{\mathcal T}'$}\label{intr11}

The submodules of $\sG=\mathbb Z_{p^2}^2$ that can be gauged in the bulk are generated by
\be \begin{pmatrix}
	p & 0\\0 & p
\end{pmatrix},~~\begin{pmatrix}
	0 & p
\end{pmatrix},~~\begin{pmatrix}
	p & ap
\end{pmatrix},~~a\in\mathbb Z_p.\ee
If $p=1\mod 4$ there are two additional modules that can be gauged, generated by
\be \begin{pmatrix}
	1 & \pm x
\end{pmatrix},~~x^2=-1.\ee
In each case, we build the diagonal interfaces $\mathcal T$ and $\mathcal T'=\mathcal T/\H$. The complete set of interfaces can then be obtained by fusing these diagonal interfaces with surface operators of $\mathcal T$.
\begin{enumerate}
	\item \textbf{Case 1:} $\H=\begin{pmatrix}
		p & 0\\
		0 & p
	\end{pmatrix}$. the resulting $1$-form symmetry group $\H^\perp/\H$ is trivial and the Hilbert space of $\mathcal T/\H$ is one-dimensional. The interface is given by
	\be \mathcal I|0\rangle_{\mathcal T/\H}={1\over p}\sum_{a,b\in\mathbb Z_p}|ap,bp\rangle.\label{inter1}\ee
	Since the right sector is trivial, we will ignore it. For the left sector, if $p=3\mod 4$, there is a single Lagrangian, given by 
	\be L_p=\begin{pmatrix}
		p & 0\\
		0 & p
	\end{pmatrix},\ee
	hence the gravity state is proportional to the corresponding stabilizer state.
	
	For $p=1\mod 4$ there is an additional orbit with two Lagrangian submodules
	\be L_{\pm}=\begin{pmatrix}
		1 & \pm x
	\end{pmatrix},\quad x^2=-1.\ee
	Denoting the corresponding stabilizer states by $|L_p\rangle,|L_\pm\rangle$, the genus-1 TQFT gravity partition function is
	\be \mathcal Z[\mathcal T,\Sigma_1;w]={pw_1\over 1+p}(|L_+\rangle+|L_-\rangle)+\left({w_1(p-1)\over p+1}+w_p p\right)|L_p\rangle .\ee
	We can also calculate the genus-$g$ gravity partition function by restricting to the diagonal weights $\tilde w_{p^i}=w_{(p^i,\dots,p^i)}$
	\be \mathcal Z[\mathcal T,\Sigma_g;\tilde w]={p^g\tilde w_1\over 1+p^g}(|L_+\rangle+|L_-\rangle)+\left({\tilde w_1(p^g-1)\over p^g+1}+\tilde w_p p^g\right)|L_p\rangle .\ee
	\item \textbf{Case 2:} $\H\cong\mathbb Z_{p^2}$. These isotropic subgroups exist only when $p=1\mod 4$. They are given by $(1,\pm x)$, where $x^2=-1$. The Hilbert space of $\mathcal T/\H$ is one-dimensional and the interface is given by
	\be \mathcal I|0\rangle_{\mathcal T/\H}={1\over p}\sum_{a\in\mathbb Z_{p^2}}|a,\pm ax\rangle .\label{inter2}\ee
	The resulting TQFT gravity state is the same as in case 1, up to an overall constant.
	\item \textbf{Case 3:} $\H\cong\mathbb Z_{p}$. There are $p+1$ choices of $\H$ given by the generators $\begin{pmatrix}
		p & ap
	\end{pmatrix}$ and $\begin{pmatrix}
		0 & p
	\end{pmatrix}$. The $1$-form symmetry group of $\mathcal T/\H$ is $\mathbb Z_{p^2}$. The interfaces are respectively
	\be \mathcal I|x\rangle_{\mathcal T/\H}={1\over \sqrt p}\sum_{y\in\mathbb Z_p}|yp-ax,yap+x\rangle,~~a\in\mathbb Z_p ,\ee
	\be \mathcal I|x\rangle_{\mathcal T/\H}={1\over \sqrt p}\sum_{y\in\mathbb Z_p}|x,yp\rangle .\label{int3}\ee
	We will only consider $\H=(0,p)$, as the other choices are similar. Fusing (\ref{int3}) with the surface operators, we obtain $2$ orbits. The orbit $\O_1$, with elements isomorphic to $\mathbb Z_{p^2}\times\mathbb Z_p$ and $\O_p$, containing the unique Lagrangian submodule isomorphic to $\mathbb Z_p^3$.
	Define the states $|\O_1\rangle$ and $|\O_2\rangle$ to be the uniform average of each orbit divided by the cardinality of the orbit, as in (\ref{Oavg}). Then, the genus-$1$ TQFT gravity partition function takes the following simple form
	\be \mathcal Z[\mathcal T\times\bar{\mathcal T}',\Sigma_1;w]=p\left(w_1|\O_1\rangle+w_2|\O_2\rangle\right).\ee
	We can also calculate the genus-$g$ partition function by restricting to diagonal weights $\tilde w_{p^i}=w_{(p^i,\dots,p^i)}$, to obtain
	\be \mathcal Z[\mathcal T\times\bar{\mathcal T}',\Sigma_g;\tilde w]=\tilde w_1{(p+1)p^g\over 1+p^g}|\O_1\rangle+\left(\tilde w_1{p^g-p\over p^g+1}+\tilde w_2 p^g\right)|\O_2\rangle.\ee
	
\end{enumerate}
 
 \subsection{Example III: $\sG=\mathbb Z_{p^2}\times\mathbb Z_p$}\label{sec:ex3}
 
 In this example, we consider the TQFT $\mathcal T\times\bar{\mathcal T}$, where $\mathcal T=(\sG,q)$ with $\sG=\mathbb Z_{p^2}\times\mathbb Z_p$ and $q(a,b)=a^2+pb^2$. We also consider the theory $\mathcal T\times\bar{\mathcal T}'$, where $\mathcal T'=\mathcal T/\H$ for the isotropic subgroup $\H\cong\mathbb Z_p$ of $\sG$.
 
  \subsubsection{Surface operators}
 To construct the surface operators, we gauge $\H\subseteq\G$ on a surface with a choice of discrete torsion. There are $4p+2$ possible choices, classified as follows
 \begin{enumerate}
 	\item $\H=\{0\}$. The corresponding surface operator is the trivial one.
 	\item $\H\cong\mathbb Z_p$. There are $p+1$ subgroups in this class $(0,1)$, $(p,0)$ and $(p,a)$ with $a\in \{1,\dots,p-1\}$. The corresponding Lagrangians are
 	\be \begin{pmatrix}
 		1 & 0 & 1 & 0\\
 		0 & 1 & 0 &-1
 	\end{pmatrix},~~\begin{pmatrix}
 		p & 0 & p & 0\\
 		0 & 1 & 0 & 1\\
 		p & 0 & -p & 0
 	\end{pmatrix},~~\begin{pmatrix}
 		-a & 1 & -a & 1\\
 		p & a & -p & -a
 	\end{pmatrix},\quad a\in\{1,\dots,p-1\}.\ee
 	\item $\H\cong\mathbb Z_{p^2}$. There are $p$ subgroups of this form, generated by $(1,a)$ for $a\in\mathbb Z_p$ and the Lagrangians are
 	\be \begin{pmatrix}
 		ap & -1 & ap & -1\\
 		1 & a & -1 & -a
 	\end{pmatrix},\quad a\in\mathbb Z_p.\ee
 	
 	\item $\H\cong\mathbb Z_p\times\mathbb Z_p$. There is $1$ subgroup, with generators $\langle (p,0),(0,1)\rangle$ and $p$ choices of discrete torsion $H^2(\mathbb Z_{p}^2,\mathbb C^\times)=\mathbb Z_p$. In total there are $p$ Lagrangians in this class given by
 	\be \begin{pmatrix}
 		p & 0 & p & 0\\
 		p& 0 & -p &0\\
 		0 & 1 & 0 & -1		
 	\end{pmatrix},~~\begin{pmatrix}
 		p& -a & -p &-a\\
 		a & 1 & a & -1		
 	\end{pmatrix},\quad a\in\{1,\dots,p-1\}.\ee
 	\item $\H=\G$. There is a unique subgroup, but $p$ choices of discrete torsion, for a total of $p$ Lagrangians given by
 	\be \begin{pmatrix}
 		1 & -a & -1 & -a\\
 		ap & 1 & ap & -1
 	\end{pmatrix},\quad a\in\mathbb Z_p.\ee
 \end{enumerate}

 The Lagrangian submodules of $\G$ are split into 2 orbits under the action of $\O(\G)$. The first orbit $\O_{(1,1)}$ contains the $4p$ invertible surface operators, with Lagrangians isomorphic to $\mathbb Z_{p^2}\times\mathbb Z_p$. The second orbit $\O_{(p,1)}$ consists of $2$ non-invertible surfaces whose Lagrangians are isomorphic to $\mathbb Z_p^3$.

 \subsubsection{Gravity state of $\mathcal T\times\bar{\mathcal T}$}
 
 In this theory the surface operators belong to $|\mathbb S|=2$ orbits with $\mathbb S=\{(1,1),(p,1)\}$. 
 Meawhile, there are $|\mathbb T|=g+1$ symplectic orbits at genus $g$. The matrix $T_{\vec q;\vec q'}$ with the ordering $(1,1),(p,1)$ for its indices, reads
 \be T={2p^{2g}\over p^{3g}+p^{1+2g}-p^g-p}\begin{pmatrix}
 	2p^{g+1} & -2p\\
 	-2p & p^g+p^{1-g}+p-1
 \end{pmatrix}.\ee
 
 \paragraph{Genus 1:} At $g=1$ there are two symplectic orbits with $\vec d\in\{1,p\}$. The seed vector is $v_{p^i}^{p^j}=p^{\min\{i,j\}}$, and we find that the $\t$-matrix is diagonal
 \be \t={2p^2\over p+1}\begin{pmatrix}
 	1 & 0\\
 	0 & 1
 \end{pmatrix}.\ee
 This leads to the TQFT gravity partition function
 \be \begin{split}\mathcal Z[\mathcal T\times\bar{\mathcal T},\Sigma_1;w]&={w_12p^2\over 1+p}|\overline{\mathcal L_{(1,1)}}\rangle+{w_p2p^2\over p+1}|\overline{\mathcal L_{(p,1)}}\rangle\\&={w_1p\over 2(1+p)}\sum_{Q\in \O_{(1,1)}}|\mathcal L_{(1,1),Q}\rangle+{w_pp^2\over p+1}\sum_{Q\in\O{(p,1)}}|\mathcal L_{(p,1),Q}\rangle.\end{split}\ee
 
 \paragraph{Genus $>$ 1:} At $g>1$ we have $|\mathbb T|>|\mathbb S|$, therefore $\t$ is not uniquely defined and we restrict again to the diagonal weights $\tilde w_i\equiv w_{p^i,\dots,p^i}$. The seed vector restricted to this subspace is $\bar v_{p^i}^{p^j}=p^{g\min\{i,j\}}$, and we find the restricted matrix $\bar \t$
 \be \bar \t={p^{2g}\over (1+p^g)(p+p^2)}\begin{pmatrix}
 	1 & 1-p^{1-g}\\
 	0 & p+p^g
 \end{pmatrix}.\ee
 The TQFT gravity partition function is
 \be \mathcal Z[\mathcal T\times\bar{\mathcal T},\Sigma_g;\tilde w]={p^{2g}\over (1+p^g)(p+p^2)}(\tilde w_1|\overline{\mathcal L_{(1,1)}}\rangle+(\tilde w_1(1-p^{1-g})+\tilde w_p(p+p^g))|\overline{\mathcal L_{(p,1)}}\rangle).\ee
 
  \subsubsection{Interfaces and gravity state of $\mathcal T\times\bar{\mathcal T}'$}
 
 There is a unique isotropic submodule $\H= (p,0)$ of $\sG$. The theory $\mathcal T/\H$ has $1$-form symmetry $\K=\H^\perp/\H\cong \mathbb Z_p$ and quadratic form $q'(a)=a^2\mod p$. The interface $\mathcal I:\mathcal H_{\mathcal T/\H}\to\mathcal H_{\mathcal T}$ reads
 \be \mathcal I={1\over \sqrt{p}}\sum_{a,b\in\mathbb Z_p} |pb,a\rangle\langle a|_{\mathcal T/\H} .\label{intr1}\ee
The Lagrangian submodules of $\G=\mathbb Z_{p^2}\times \mathbb Z_p\times\bar{\mathbb Z}_p$ can be obtained by fusing the interface (\ref{intr1}) with surface operators of $\mathcal T$. This way, we obtain two submodules, belonging to the same orbit of $\O(\G)$
 \be \lambda_\pm=\begin{pmatrix}
 	0 & 1 & \pm1\\
 	p & 0 & 0
 \end{pmatrix}.\ee
 Hence, the TQFT gravity path integral at any genus is simply proportional to the uniform sum of these two states
 \be \mathcal Z[\mathcal T\times\bar{\mathcal T}',\Sigma_g;\tilde w]\propto|\lambda_+\rangle+|\lambda_-\rangle.\ee
 
 \subsection{Example IV: $\G=\mathbb Z_{p^2}^3$}
In this final example, we consider a theory $\mathcal T$ with $1$-form symmetry group $\G=\mathbb Z_{p^2}^3$ and anti-diagonal bilinear form
 \be \g=\begin{pmatrix}
 	0 & 0 & 1\\
 	0 & 1 & 0\\
 	1 & 0 & 0
 \end{pmatrix}.\ee
 We first check whether the higher central charges vanish \cite{Kaidi:2021gbs}. It is straightforward to calculate
 \be e^{{2\pi i\over 8}c_n}={1\over p^3}\sum_{a,b,c\in\mathbb Z_{p^2}} \omega^{n(2ac+b^2)}=1,\quad\forall n:\gcd(n,p^6/\gcd(n,p^6)),\ee
 where we used that $p^2=1\mod 4$ for any odd $p$.
 
The non-trivial step is finding all the Lagrangian submodules of $\G$, which can be done by explicitly solving the Lagrangian conditions for this group. We find $p+1$ submodules isomorphic to $\mathbb Z_{p^2}\times\mathbb Z_p$ and one submodule isomorphic to $\mathbb Z_p^3$
 \be \mathcal L_{p+1}=\begin{pmatrix}
 	p & 0 & 0\\
 	0 & p & 0\\
 	0 & 0 & p
 \end{pmatrix},~~\mathcal L_p=\begin{pmatrix}
 	0 & 0 & 1\\
 	0 & p & 0
 \end{pmatrix},\ee
 \be \mathcal L_a=\begin{pmatrix}
 	1 & a & -2^{-1}a^2\\
 	0 & p & -pa
 \end{pmatrix},~~a=0,1,\dots,p-1.\ee

 The intersection numbers are $|\mathcal L_i\cap \mathcal L_j|=p$ for $0\leq i<j\leq p$, while $|\mathcal L_i\cap \mathcal L_{p+1}|=p^2$ for $0\leq i\leq p$. The inverse of the intersection matrix is
 \be (D)^{-1}={1\over p^g-p^{-g}}(p^g\id+A),\ee
 where
 \be A=\begin{pmatrix}
 	0 & 0& \cdots &0 & -1\\
 	0 & 0& \cdots &0 & -1\\
 	\vdots& \vdots &\vdots &\vdots &\vdots \\
 	0 & 0& \cdots &0 & -1\\
 	-1 & -1& \cdots &-1 & p^{1-g}\end{pmatrix}.\ee
 The matrix $Y$ reads
 \be Y={p+1\over p^g-p^{-g}}\begin{pmatrix}
 	p^g & -1\\
 	-1 &{p^g+ p^{1-g}\over p+1}
 \end{pmatrix} .\ee
 
 \paragraph{Genus 1:}
At $g=1$, the $\t$-matrix is diagonal and reads
 \be \t= \begin{pmatrix}
 	p & 0\\
 	0 & p
 \end{pmatrix}.\ee
 Therefore the gravity state at genus $1$ is
 \be \mathcal Z[\mathcal T,\Sigma_1;w]=p\left({w_1\over p+1}\sum_{i=0}^p |\mathcal L_i\rangle+w_p|\mathcal L_{p+1}\rangle\right).\ee
 
  \paragraph{Genus$>$1:}
  At $g>1$, restricting to the diagonal weights $\tilde w_{p^i}=w_{(p^i,\dots,p^i)}$ we obtain the restricted matrix
  \be \bar\t=\begin{pmatrix}
  	p^g{p+1\over p^g+1} & {p^g-p\over p^g+1}\\
  	0 & p^g
  \end{pmatrix} .\ee
  This leads to the following genus-$g$ gravity state
  \be \mathcal Z[\mathcal T,\Sigma_g;\tilde w]={\tilde w_1p^g\over 1+p^g}\sum_{i=0}^p |\mathcal L_i\rangle+\left(\tilde w_p p^g+\tilde w_1{p^g-p\over 1+p^g}\right)|\mathcal L_{p+1}\rangle.\ee

\section{Outlook}\label{sec:conc}

In this work, we defined the gravitational path integral of a 3D abelian TQFT $\mathcal T$ as a sum over all topologies with genus-$g$ boundary $\Sigma_g$. The path integral of $\mathcal T$ prepares an $\eta$-Lagrangian stabilizer state on any single topology with boundary $\Sigma_g$, partitioning the topologies into finitely many equivalence classes. Thus, the gravitational path integral can be equivalently defined as a weighted sum over these stabilizer states. The gravitational path integral can then be expressed as a weighted sum over $\q$-Lagrangian stabilizer states, which describe topological boundary conditions on $\Sigma_g$.
From the boundary perspective, this amounts to an average over partition functions of two-dimensional CFTs, specifically a finite ensemble of Narain CFTs, where the chiral central charge is a multiple of $8$\footnote{When the chiral central charge is not a multiple of 24, the framing anomaly of the bulk manifold can be canceled by stacking copies of $(E_8)_1$.}.

We introduced the $\t$-matrix, relating bulk and boundary weights and developed a method to calculate its entries from the TBCs that the TQFT admits. Specifically, the $\t$-matrix is given by
\be \t_{\vec d,\vec q}=\sum_{\vec q';\vec q}v_{\vec q'}^{\vec d}Y_{\vec q';\vec q},\ee
where the matrix $Y_{\vec q';\vec q}$ is determined by the overlaps between orbits of TBC states on $\Sigma_g$ and $v_{\vec q}^{\vec d}$ is a matrix that can be determined only from the $1$-form symmetry group $\G$ and the genus of the boundary surface.
Above, the index $\vec d$ denotes a symplectic orbit of representative topologies, while $\vec q,\vec q'$ denote orthogonal orbits of TBCs.
 We illustrated this calculation with several examples.
 
In the process, we also described (section \ref{sec3}) a systematic method for constructing all TBC of abelian TQFTs of the form $\mathcal T\times\bar{\mathcal T}$, by taking advantage of their correspondence with surface operators of $\mathcal T$, as well as the groupoid structure of surface operators \cite{Gaiotto_2021}. We further generalized this to some theories of the form $\mathcal T \times \overline{\mathcal T}'$, where the classification of TBCs is equivalent to the classification of interfaces between the two sectors.

We conclude with a list of open problems.

\paragraph{The $\t$-matrix} In this paper we computed the $\t$-matrix in explicit examples using a brute-force approach, which requires the full set of TBCs. Although we described a systematic procedure for constructing this set, it becomes impractical for groups of large length. Interestingly, the $\t$-matrices we obtained have a simple structure. In fact, only overlaps between orthogonal orbits of TBC states are needed, rather than overlaps between all states individually. Alternatively, the construction of section \ref{sec:TQFTgrav} can be reformulated in terms of overlaps between symplectic orbits. These two observations suggest that a simpler method to compute $\t$ may exist.

\paragraph{Universality of bulk weights} In this work we largely remained agnostic about the choice of bulk weights. Consistency with the holographic principle requires that the boundary weights are independent of the genus of the surface, but this does not uniquely fix the bulk weights. A prescription consistent with holography was proposed in \cite{Dymarsky:2024frx}, but it is unclear whether other consistent prescriptions exist. Alternatively, one could define the TQFT gravity path integral by starting from a (formal) sum over all distinct bulk topologies with boundary $\Sigma_g$, weighted by a measure $\mu$, as in \cite{Nicosanti:2025xwu}. For consistency with a path-integral interpretation, $\mu$ should be TQFT-independent. Since the TQFT path integral on any topology yields an $\eta$-Lagrangian stabilizer state, this construction reduces to (\ref{gravityZdef}), with $w_{\vec d}$ determined by $\mu$, possibly after regularization. An open question is what constraints, if any, the universality of $\mu$ imposes on the bulk weights $w_{\vec d}$.

\paragraph{Semi-classical gravity} The handlebody topologies $\vec 1 \in \mathbb T$ can be interpreted as the semi-classical contributions to the gravitational path integral, which are always present. The non-handlebody topologies $\vec d \neq \vec 1$ then play the role of quantum corrections. The semi-classical limit corresponds to sending either the exponent or the length of $\G$ to infinity, in which case we expect the contributions from non-handlebody topologies to be suppressed. Indeed, the prescription in \ref{DSpresc} along with (\ref{vmatr}) imply that the bulk weights $w_{\vec d}$ with $\vec d\neq\vec 1$ are suppressed as the exponent $N$ is increased. Our approach is not straightforward to apply to groups of large length, but one should instead reformulate section \ref{sec:TQFTgrav} using the overlap matrix of $\eta$-Lagrangian submodules. It would be interesting to rigorously and generally prove this suppression of non-handlebody contributions.

\paragraph{Generalizations} Several extensions of this work remain. While we focused on groups of odd exponent and specific classes of quadratic forms, the construction should generalize to all abelian TQFTs that admit TBCs, though the details deserve careful treatment. Another direction is to include correlation functions. In this case the boundary surface has punctures, enlarging the mapping class group action and requiring a generalized averaging procedure \cite{Romaidis:2023zpx}. Finally, extending our framework to non-abelian TQFTs is another interesting direction. The approach of section \ref{sec:TQFTgrav} can still be used if one can identify a spanning set of $\Sp(2g,\mathbb Z)$-invariant states, however the TBC states (corresponding to physical modular invariants) do not generally span this space. From the holographic perspective, the gravitational state must be expressed entirely in terms of physical invariants, which means the bulk weights must be tuned so that the resulting state lies in their span.

\section*{Acknowledgements}

The author is grateful to Chen Fei for contributing the material in Appendix B and for valuable discussions, and to Anatoly Dymarsky for helpful comments and discussions.

\newpage
\appendix

\section{Higher symmetry lines and Wigner functions}\label{highergauging}

In this appendix we review how to construct higher symmetry lines on the surface operator \cite{Roumpedakis:2022aik}. These ``higher lines" embedded on the surface are the phase-space point operators in the stabilizer formalism. These operators form a very useful alternative orthonormal basis for $\End(\mathcal H)$. They are typically defined by the symplectic Fourier transform of the Heisenberg-Weyl operators \cite{Gross_2006}
\be \hat L^{(g)}(\alpha)={1\over |\G|^g}\sum_{\beta\in \G^{2g}} \omega^{J(\alpha,\beta)}\W(\beta).\label{pso}\ee
We can interpret (\ref{pso}) as lines on a surface operator obtained by gauging $\G$ with trivial torsion on a genus-$g$ surface inserted parallel to the two boundaries of $\Sigma_g\times I$. This surface is the charge-conjugation operator
\be \rho_{\G,0}={1\over |\G|^g}\sum_{\alpha\in \G\otimes H_1(\Sigma_g,\mathbb Z)}\W(\alpha)=\left(\sum_{x\in \G}|x\rangle\langle -x|\right)^{\otimes g} .\label{gaugG}\ee
Gauging the group $\G$ results in a ``quantum symmetry" $\what\G\cong \G$ of lines confined to the surface. These lines can be constructed by bringing together two bulk lines of opposite orientation and same charge and fusing them with the surface from opposite sides
\be \hat L^{(g)}(\alpha)=\W(\alpha)\rho_{\G,0}\W(\alpha)^\dagger.\ee
On a torus, they read explicitly
\be \hat L^{(1)}(a,b)=\sum_{x\in\G} \omega^{2bx}|x+a\rangle\langle-x+a|.\ee
The expansion coefficients of an operator $\rho$ in the $\hat L(\alpha)$ basis is the Wigner function
\be \mathbb W_{\rho}(\alpha)\equiv {1\over |\G|^g}\tr(\hat L(\alpha)^\dagger \rho).\ee
The Wigner function has several important properties \cite{Gross_2006}. One of the most important properties we make use of is that the Wigner function of a stabilizer state is nonnegative.

More generally, we can construct ``higher" line operators by performing the Fourier transform over a subgroup $\H$ and adding discrete torsion $\nu\in H^2(\H,\mathbb C^\times)$. This process constructs line operators labeled by $\what\H$ and confined on a surface
\be \hat L_{\H,\nu}^{(g)}(\alpha)={1\over |\H|^g}\sum_{\beta\in \H^{2g}}\omega^{J(\alpha,\beta)}\nu(\beta)\W(\beta),~~~\alpha\in\what\H^{2g} .\ee
The higher symmetry lines belonging to $\what\H$ can be obtained by fusing two bulk lines with charges in $\H$ and opposite orientation, from opposite sides of the surface
\be \hat L_{\H,\nu}^{(g)}(\alpha)=\W(\alpha)\rho^{\otimes g}_{\H,\nu}\W(\alpha)^\dagger,~~\alpha\in\what\H^{2g}.\ee
Note that in the above definition $\hat L_{\H,\nu}(\alpha)=\rho_{\H,\nu}$ if $\alpha\in (\G/\H)^{2g}$.
Similar to the Heisenberg-Weyl operators, they can be written as a tensor product of torus operators
\be \hat L^{(g)}(\alpha)=\hat L(\alpha_1)\otimes\cdots\otimes \hat L(\alpha_g).\ee
We can write the torus operators explicitly
\be \hat L_{\H,\nu}(a,b)=\W(a,b)\rho_{\H,\nu}\W(a,b)^\dagger=\sum_{(x,h)\in \H\times\H^\perp}\omega^{-2\g(b,x)}|h+a+x(-\id+B\g^{-1})\rangle\langle h+a+x(1+B\g^{-1})| .\ee
Above $a,b\in\H$, since $\hat L_{\H,\nu}(a,b)=\rho_{\H,\nu}$ for $a,b\in \H^\perp$.

If $\H\neq \G$, there is another set of symmetry lines parametrized by $\G/\H$. This is the residual symmetry after gauging $\H$. They can be constructed by attaching lines of the same charge and orientation from the two opposite sides of the surface
\be \tilde L_{\H,\nu}(a,b)=\W(2^{-1}a,2^{-1}b)\rho_{\H,\nu}\W(2^{-1}a,2^{-1}b).\ee

Unless $\what\H\times\G/\H\cong\G$, the symmetry lines generally have a mixed anomaly. Subgroups of either $\what\H$ or $\G/\H$ are non-anomalous and can be gauged separately, but mixed subgroups cannot. 

Due to the fact that the surface operators are invariant under the mapping class group, the higher line operators transform covariantly under $\Sp(2g,\mathbb Z_N)$
\be U_\gamma L_{\H,\nu}^{(g)}(\alpha,\beta)U_\gamma^\dagger=L_{\H,\nu}^{(g)}(\alpha\gamma,\beta\gamma),~~~\gamma\in \Sp(2g,\mathbb Z_N).\ee
As a consequence the Wigner function transforms contravariantly
\be \mathbb W_{U_\gamma\rho U_\gamma^\dagger}(\alpha\gamma)=\mathbb W_{\rho}(\alpha).\ee

\section{The minimal matrix of a $\q$-Lagrangian submodule}\label{appB}

 Without loss of generality, we set $N=p^m$.
\begin{definition}
	Let $\mathcal L$ be a Lagrangian submodule of $\G=\mathbb Z_{p^m}^{2n}$ with respect to $\q$ associated with the bilinear form (\ref{antidiag}).  
	A \emph{minimal matrix} $M$ of $\mathcal L$ is an $n\times 2n$ matrix which, up to permutation of columns, satisfies the following conditions:
	\begin{enumerate}
		\item Each row of $M$ lies in $\mathcal L$.
		\item $M_{ij}=0$ whenever $j<i$.
		\item The entries of the $i$-th row have greatest common divisor $p^{k_i}$.
		\item The sequence of exponents satisfies $0\leq k_1 \leq \cdots \leq k_n \leq \lfloor m/2 \rfloor$.
		\item The sequence $(k_1,\dots,k_n)$ is minimal among all choices of matrices satisfying the above.
	\end{enumerate}
\end{definition}
\begin{lemma}
	A minimal matrix $M$ always exists, and it uniquely determines the submodule $\mathcal L$.
\end{lemma}
\begin{proof}
	We first show existence of $M$.
	From the definition, $p^{k_1}$ is the greatest common divisor of the entries of all elements of $\mathcal L$.  
	Since $\mathcal L$ is Lagrangian, $k_1 \leq \lfloor m/2 \rfloor$. After a $\q$-preserving permutation of columns if necessary, there exists a vector $v_1 \in \mathcal L$ of the form
	$
	v_1=(p^{k_1},*,\dots,*).
	$
	Define the submodule
	$
	\mathcal L'=\{v\in \mathcal L : v=(0,x_2,\dots,x_{2n})\}.
	$
	Then $\mathcal L = \mathcal L' + \langle v_1\rangle$.  
	For each $v=(0,x_2,\dots,x_{2n})\in\mathcal L'$, define
	$
	\bar v = (\bar x_2,\dots,\bar x_{2n-1})\in \mathbb Z_{p^{m-2k_1}}^{2n-2},
	$
	where $p^{k_1}\bar x_i=x_i(\!\!\!\!\mod p^{k-k_1})$.  
	Set
	$
	\bar{\mathcal L} = \{\bar v : v\in \mathcal L'\}.
	$
	We claim that $\bar{\mathcal L}$ is a Lagrangian submodule of $\mathbb Z_{p^{m-2k_1}}^{2n-2}$.  
	Suppose not; then there exists $u$ such that $\bar\g(u,\bar v)=0$ for all $\bar v\in\bar{\mathcal L}$, where $\bar\g$ is the bilinear form restricted to $\mathbb Z_{p^{m-2k_1}}^{2n-2}$.  
	Choose $u'=(0,y_2,\dots,y_{2n-1},0)\in\mathbb Z_{p^m}^{2n}$ such that $\bar u'=u$. Then there exists $y_{2n}$ such that
	$
	\g\big(v_1,(0,y_2,\dots,y_{2n-1},y_{2n})\big)=0,
	$
	implying $u\in \bar{\mathcal L}$, a contradiction.  
	Thus $\bar{\mathcal L}$ is Lagrangian.  	
	By induction on $(m,n)$, $\bar{\mathcal L}$ admits a minimal matrix and choosing a lift produces a minimal matrix $M$ for $\mathcal L$.
	
	We now show that $M$ uniquely determines $\mathcal L$. By induction, it suffices to show that the vector
	$
	(0,\dots,0,p^{m-k_1}) \in \mathcal L.$
	This is true since from the definition it follows that $p^{k_1}$ is the gcd of entries of all elements of $\mathcal L$.  
	Hence $\mathcal L$ is uniquely determined by $M$.
\end{proof}

\bibliographystyle{unsrt} 
\bibliography{../bibt} 

\end{document}